%% file: no-colours-and-ids.tex
\documentclass{article}
\usepackage{fullpage} 

\usepackage[
normalsections, 
normalmargins, 
normallists, 
normalfloats, 
normalindent, 
normaltitle, 
normalleading, 
normallooseness, 
]{savetrees}

\usepackage{amssymb,comment}
\usepackage{amsmath}
\usepackage{tikz}
\usetikzlibrary{snakes}
\usepackage{times}
\usepackage{amstext}
\usepackage{calc}
\usepackage{amsopn}
\usepackage[noend]{algorithmic}
\usepackage[boxed]{algorithm}
\usepackage{eucal}
\usepackage{latexsym}
\usepackage{tabularx}

\usepackage{amsthm}
\usepackage{xspace}

\newtheorem{theorem}{Theorem}[section]
\newtheorem{lemma}[theorem]{Lemma}
\newtheorem{corollary}[theorem]{Corollary}

\newtheorem{conjecture}[theorem]{Conjecture}
\theoremstyle{definition}
\newtheorem{definition}[theorem]{Definition}

\newcommand{\colouring}{{\sc{Colouring}}\xspace}
\newcommand{\cliquecover}{{\sc{Edge Clique Cover}}\xspace}
\newcommand{\icliquecover}{{\sc{Compression Clique Cover}}\xspace}
\newcommand{\kwaycut}{{\sc{$k$-Way Cut}}\xspace}
\newcommand{\clique}{{\sc{Clique}}\xspace}
\newcommand{\dirmc}{{\sc{Directed Multiway Cut}}\xspace}
\newcommand{\dirvmc}{{\sc{Directed Vertex Multiway Cut}}\xspace}
\newcommand{\diremc}{{\sc{Directed Edge Multiway Cut}}\xspace}
\newcommand{\prepdirvmc}{{\sc{Promised Directed Vertex Multiway Cut}}\xspace}
\newcommand{\prepmc}{{\sc{Promised Multiway Cut}}\xspace}
\newcommand{\mc}{{\sc{Multiway Cut}}\xspace}
\newcommand{\multicut}{{\sc{Multicut}}\xspace}
\newcommand{\emulticut}{{\sc{Edge Multicut}}\xspace}
\newcommand{\vmulticut}{{\sc{Vertex Multicut}}\xspace}

\newcommand{\prel}{\mathcal{R}}
\newcommand{\Ii}{\mathcal{I}}
\newcommand{\Hh}{\mathcal{H}}

\newcommand{\unlesscompass}{\ensuremath{\textrm{NP} \subseteq \textrm{coNP}/\textrm{poly}}\xspace}

\newcommand{\defproblemnoparam}[3]{
  \vspace{1mm}
\noindent\fbox{
  \begin{minipage}{\textwidth}  
  #1 \\ 
  {\bf{Input:}} #2  \\
  {\bf{Task:}} #3
  \end{minipage}
  }
  \vspace{1mm}
}

\begin{document}

\title{Clique cover and graph separation: New incompressibility results}
\author{
  Marek Cygan\thanks{Institute of Informatics, University of Warsaw, Poland, \texttt{cygan@mimuw.edu.pl}.} \and
  Stefan Kratsch\thanks{Utrecht University, Utrecht, the Netherlands, \texttt{s.kratsch@uu.nl}.} \and
	Marcin Pilipczuk\thanks{Institute of Informatics, University of Warsaw, Poland, \texttt{malcin@mimuw.edu.pl}.} \and
  Micha\l{} Pilipczuk\thanks{Department of Informatics, University of Bergen, Norway, \texttt{michal.pilipczuk@ii.uib.no}.} \and
  Magnus Wahlstr\"{o}m\thanks{Max-Planck-Institute for Informatics, Saarbr\"{u}cken, Germany, \texttt{wahl@mpi-inf.mpg.de}.}
}
\date{}

  \maketitle

\begin{abstract}
The field of kernelization studies polynomial-time preprocessing
routines for hard problems in the framework of parameterized complexity.
Although a framework for proving kernelization lower bounds
has been discovered in 2008 and successfully applied multiple times
over the last three years, establishing kernelization complexity of many
important problems remains open.
In this paper we show that, unless \unlesscompass{} and the polynomial hierarchy collapses
up to its third level, the following parameterized problems do not
admit a polynomial-time preprocessing algorithm that reduces
the size of an instance to polynomial in the parameter:
\begin{itemize}
\item \cliquecover{}, parameterized by the number of cliques,
\item \textsc{Directed Edge/Vertex} \mc{}, parameterized by the size of the cutset, even in the case of two terminals,
\item \textsc{Edge/Vertex} \multicut{}, parameterized by the size of the cutset, and
\item \kwaycut{}, parameterized by the size of the cutset.
\end{itemize}
The existence of a polynomial kernelization for \cliquecover{} was a seasoned veteran in open problem sessions.
Furthermore, our results complement very recent developments 
in designing parameterized algorithms for cut problems by Marx and Razgon~[STOC'11], 
Bousquet et al.~[STOC'11], Kawarabayashi and Thorup~[FOCS'11] and Chitnis et al.~[SODA'12].
\end{abstract}

\input{introduction}

\input{preliminaries}

\input{clique-cover-full}

\input{directed-multiway-full}

\input{multicut-full}

\input{multicut2}

\input{k-way-cut-full}

\input{conclusions}

\bibliographystyle{plain}
\bibliography{no-colours-and-ids}

\end{document}

%% file: introduction.tex
\section{Introduction}

In order to cope with the NP-hardness of many natural combinatorial problems, various algorithmic paradigms such as brute-force, approximation, or heuristics are applied. However, while the paradigms are quite different, there is a commonly used opening move of first applying polynomial-time preprocessing routines, before making sacrifices in either exactness or runtime.
The aim of the field of kernelization is to provide a rigorous mathematical
framework for analyzing such preprocessing algorithms. One of its core features is to provide quantitative performance guarantees for preprocessing via the framework of parameterized complexity, a feature easily seen to be infeasible in classical complexity (cf.~\cite{HarnikN10}).


In the framework of parameterized complexity an instance~$x$ of a 
parameterized problem comes with an integer parameter~$k$.
A {\em{kernelization algorithm}} ({\em{kernel}} for short) is a polynomial
time preprocessing routine that reduces the input instance~$x$
with parameter~$k$ to an equivalent instance of size bounded by~$g(k)$ for some computable function~$g$.
If~$g$ is small, after preprocessing even an exponential-time brute-force algorithm might
be feasible. Therefore small kernels, with~$g$ being linear or polynomial,
are of big interest.

Although polynomial kernels for a wide range of problems have been developed for the
last few decades (e.g.,~\cite{hitting-set:kernel,meta-kernelization,buss-goldsmith,moje:povd,fomin:bidim-kernels,vc:2k,fvs:quadratic-kernel}; see also the surveys of Guo and Niedermeier~\cite{guo:survey} and Bodlaender~\cite{Bodlaender09}), a framework for proving
kernelization lower bounds was discovered only three years ago
by Bodlaender et al.~\cite{bodlaender:kernel},
with the backbone theorem proven by Fortnow and Santhanam~\cite{fortnow:kernel}.
The crux of the framework is the following idea of a composition. Assume we are
able to combine in polynomial time an arbitrary number of instances~$x_1,x_2,\ldots,x_t$
of an NP-complete problem~$L$ into a single instance~$(x,k)$ of a parameterized
problem~$Q \in NP$ such that~$(x,k) \in Q$ if and only if one of the instances~$x_i$ is in~$L$, while~$k$ is bounded polynomially in~$\max_i |x_i|$. If such a {\em{composition}}
algorithm was pipelined with a polynomial kernel for the problem~$Q$,
we would obtain an OR-distillation of the NP-complete language~$L$:
the resulting instance is of size polynomial in~$\max_i |x_i|$,
 possibly significantly smaller than~$t$,
but encodes a disjunction of all input instances~$x_i$ (i.e., an OR-distillation is a compression of the logical OR of the instances).
As proven by Fortnow and Santhanam~\cite{fortnow:kernel}, existence of
such an algorithm would imply~\unlesscompass{}, which is known to cause
a collapse of the polynomial hierarchy to its third level~\cite{nokernel-collapse2,nokernel-collapse1}.

The astute reader may have noticed that the above description of a composition is actually using the slightly newer notion of a cross-composition~\cite{cross-composition}. This generalization of the original lower bound framework will be the main ingredient of our proofs. 
The framework of kernelization lower bounds was also extended by Dell and van Melkebeek
\cite{dell:kernel} to allow excluding kernels of particular exponent in the polynomial.
Recently, Dell and Marx~\cite{dell-marx:soda12} and, independently, Hermelin and Wu
\cite{hermelin-wu:soda12} simplified this approach and applied it to various
packing problems.

The aforementioned (cross-)composition algorithm is sometimes called an {\em{OR-composition}},
as opposed to an {\em{AND-composition}}, where we require that the output instance
$(x,k)$ is in~$Q$ if and only if {\em{all}} input instances belong to~$L$.
Various problems have been shown to be AND-compositional, with the most important
example being the problem of determining whether an input graph has treewidth no
larger than the parameter \cite{bodlaender:kernel}. It is conjectured~\cite{bodlaender:kernel}
that no NP-complete problem admits an AND-distillation, which would be a result of pipelining
an AND-composition with a polynomial kernel. However, it is now a major open problem
in the field of kernelization to support this claim with a proof based
on a plausible complexity assumption.

Although the framework of kernelization lower bounds has been applied successfully
multiple times over the last three years (e.g.,
\cite{cross-composition,bjk:treewidth,nasze:eulerian,nasze:wg10,colours-and-ids,fernau:outtrees,stefan:ramsey,stefan:two}),
there are still many important
problems where the existence of a polynomial kernel is widely open.
The reason for this situation is that an application of the idea of a composition
(or appropriate reductions, called {\em{polynomial parameter transformations}}~\cite{bodlaender:transformations}) is far from being automatic.
To obtain a composition algorithm, usually one needs to carefully choose
the starting language~$L$ (for example, the choice of the
starting language is crucial for 
compositions of Dell and Marx~\cite{dell-marx:soda12},
and the core idea of the composition algorithms for connectivity problems
in degenerate graphs~\cite{nasze:wg10} is to use {\sc{Graph Motif}} as a starting point)
or invent sophisticated gadgets to merge the instances
(for example, the colors and IDs technique introduced by Dom et al.~\cite{colours-and-ids}
 or the idea of an instance selector, used mainly for structural parameters~\cite{cross-composition,bjk:treewidth}).

\paragraph{Our results.}

The main contribution of this paper is a proof of non-existence of polynomial kernels for four important
problems.

\begin{theorem}
\label{thm:main}
Unless \unlesscompass, \cliquecover, parameterized by the number of cliques, as well as \mc, \multicut and \kwaycut, parameterized by the size of the cutset, do not admit polynomial kernelizations.
\end{theorem}

The common theme of our compositions is a very careful choice of starting problems.
Not only do we select particular NP-complete problems, but we also restrict instances given as the input,
to make them satisfy certain conditions that allow designing cross-compositions.
Each time we constrain the set of input instances of an NP-complete problem
we need to prove that the problem remains NP-complete.
Even though this paper is about negative results, in our constructions we use
intuition derived from the design of parameterized algorithms techniques,
including iterative compression (in case of \cliquecover) introduced by Reed et al.~\cite{reed:ic} and 
important separators (in case of \multicut) defined by Marx~\cite{marx:cut}.

For the three cut problems listed in Theorem~\ref{thm:main} our kernelization hardness results complement
very recent developments in the design of algorithm parameterized by the size of the cutset~\cite{multicut-fpt1, marx:dir-mc,multicut-fpt2,  kway:thorup}.
In the following we give some motivation and related work for each of the four problems.

\paragraph{Edge clique cover.}
In the \cliquecover{} problem the goal is to cover the edges of an input graph~$G$ with at most~$k$ cliques
all of which are subgraphs of~$G$. This problem, NP-complete even in very restricted graph classes~\cite{chang-muller:cc,hoover:cc,orlin:cc}, is also known as \textsc{Covering by Cliques} (GT17), \textsc{Intersection Graph Basis} (GT59)~\cite{garey-johnson} and \textsc{Keyword Conflict}~\cite{kellerman}.
It has multiple applications in various areas in practice,
  such as computational geometry~\cite{cc-apl1},
  applied statistics~\cite{cc-apl2,cc-apl3}, and compiler optimization~\cite{cc-apl4}.
In particular, \cliquecover{} is equivalent to the problem of finding a representation
of a graph~$G$ as an intersection model with at most~$k$ elements in the universe
\cite{cliquecover:erdos,cliquecover:book,cliquecover:appl}.
Therefore, an algorithm for \cliquecover{} may be used to reveal a structure
in a complex real-world network~\cite{guillaume:cc}.
Due to its importance, the \cliquecover problem was studied from various perspectives,
including approximation upper and lower bounds~\cite{apx:cc,lund-yannakakis},
heuristics~\cite{bt:cc,cc-apl2,kellerman,kou:cc,cc-apl3,cc-apl4} and
polynomial-time algorithms for special graph classes~\cite{hoover:cc,cc-class2,cc-class1,orlin:cc}.

From the point of view of parameterized complexity,
\cliquecover was extensively studied by Gramm et al.~\cite{gghn:cc}.
A simple kernelization algorithm is known that reduces the size of the graph
to at most~$2^k$ vertices; the best known fixed-parameter algorithm
is a brute-force search on the~$2^k$-vertex kernel.
The question of a polynomial kernel for \cliquecover,
probably first verbalized by Gramm et al.~\cite{gghn:cc},
was repeatedly asked in the parameterized complexity community, for example
on the last Workshop on Kernels (WorKer, Vienna, 2011).
We show that \cliquecover is both AND- and OR-compositional (i.e., both an AND- and an OR-composition algorithm exist for some NP-complete input language~$L$),
thus the existence of a polynomial kernel would both cause a collapse of the polynomial
hierarchy as well as violate the AND-conjecture.
To the best of our knowledge, this is the first natural parameterized problem
that is known to admit both an AND- and an OR-composition algorithm.

\paragraph{Multicut and directed multiway cut.}
With \multicut and \dirmc we move on to the family of graph separation problems. 
The central problems of this area are two natural generalizations of the~$s-t$ cut problem,
namely \mc and \multicut. In the first problem we are given a graph~$G$ with
designated terminals and we are to delete at most~$p$ edges (or vertices, depending
on the variant) so that the terminals remain in different connected components.
In the \multicut problem we consider a more general setting where the input graph
contains terminal {\em{pairs}} and we need to separate all pairs of terminals.


As generalizations of the well-known~$s-t$ cut problem,
\mc and \multicut received a lot of attention in past decades.
\mc is NP-complete even for the case of three terminals~\cite{nmc-np-hardness},
thus the same holds for \multicut with three terminal pairs.
Both problems were intensively studied from the approximation perspective
\cite{mc-apx-15,multicut-apx,mc-apx-2,mc-apx-138,dir-mc-apx}.
The graph separation problems became one of the most important subareas in parameterized complexity after Marx introduced the concept of important separators~\cite{marx:cut}.
This technique turns out to be very robust, and is now a key ingredient
in fixed-parameter algorithms for various problems such as variants
of the {\sc{Feedback Vertex Set}} problem~\cite{dfvs,sfvs} or {\sc{Almost 2-SAT}}
\cite{almost2sat-fpt}.
A long line of research on \mc in the parameterized setting include~\cite{chen:mc,nmc:2k,Guillemot11a,marx:cut,razgon:arxiv2010,razgon:arxiv2011,xiao:multiway2010};
the current 
fastest algorithm runs in~$O(2^p n^{O(1)})$ time~\cite{nmc:2k}.
It is not very hard to prove that \multicut, parameterized by both the number of terminals
and the size of the cutset, is reducible to \mc~\cite{marx:cut}. Fixed-parameter tractability
of \multicut parameterized by the size of the cutset only,
after being a big open problem for a few years, was finally 
resolved positively in 2010~\cite{multicut-fpt1,multicut-fpt2}.

In directed graphs \mc is NP-complete even for two terminals~\cite{mc-apx-2}.
Very recently Chitnis et al.~\cite{marx:dir-mc} showed that \dirmc is fixed-parameter tractable.
The directed version of \multicut, parameterized by the size of the cutset,
is~$W[1]$-hard~\cite{multicut-fpt2}
(i.e., an existence of a fixed-parameter algorithm is unlikely).
The parameterized complexity of {\sc{Directed}} \multicut with fixed number of terminal pairs
or with the number of terminal pairs as an additional parameter remains open.

Although the picture of the fixed-parameter tractability of the graph separation problems
becomes more and more complete, very little is known about polynomial kernelization.
Very recently, Kratsch and Wahlstr\"{o}m came up with a genuine application
of matroid theory to graph separation problems.
They were able to obtain randomized polynomial kernels for
{\sc{Odd Cycle Transversal}}~\cite{stefan:oct}, {\sc{Almost 2-SAT}}, and
\mc and \multicut restricted to a bounded number of terminals, among others~\cite{stefan:new}.
We are not aware of any other results on kernelization of the graph separation problems.

We prove that \dirmc, even in the case of two terminals, as well as
\multicut, parameterized by the size of the cutset, are OR-compositional,
thus a polynomial kernel for any of these two problems would cause
a collapse of the polynomial hierarchy.
In fact, we give two OR-composition algorithms for \multicut:
the constructions are very different and the presented gadgets may inspire other
researchers in showing lower bounds for similar problems.

\paragraph{The $\boldsymbol{k}$-way cut problem.}
The last part of this work is devoted to another generalization of the~$s$-$t$ cut problem,
 but of a bit different
flavor. The \kwaycut{} problem is defined as follows: given an undirected graph~$G$
and integers~$k$ and~$s$, remove at most~$s$ edges from~$G$ to obtain a graph with
at least~$k$ connected components.
This problem has applications in numerous areas of computer science, such as
finding cutting planes for the traveling salesman problem,
clustering-related settings (e.g., VLSI design) or network reliability~\cite{kway:motiv}.
In general, \kwaycut{} is NP-complete~\cite{kway:first} but solvable in polynomial
time for fixed~$k$: a long line of research~\cite{kway:first,kway:improve2,kway:improve1,kway:thorup} led to a deterministic
algorithm running in time~$O(mn^{2k-2})$.
The dependency on~$k$ in the exponent is probably unavoidable:
from the parameterized perspective, the \kwaycut{} problem parameterized
by~$k$ is~$W[1]$-hard~\cite{kway:w1hard}. Moreover, the node-deletion variant
is also~$W[1]$-hard when parameterized by~$s$~\cite{marx:cut}.
Somewhat surprisingly, in 2011 Kawarabayashi and Thorup presented a fixed-parameter
algorithm for (edge-deletion) \kwaycut{} parameterized by~$s$~\cite{ken:kwaycut}.
In this paper we complete the parameterized picture of the edge-deletion \kwaycut{} problem
parameterized by~$s$ by showing that it is OR-compositional and, therefore, a polynomial kernelization
algorithm is unlikely to exist.

\paragraph{Organization of the paper.} We give some notation and formally
introduce the composition framework in Section \ref{sec:prelim}.
In subsequent sections we show compositions for the aforementioned four problems:
we consider \cliquecover in Section \ref{sec:clique-cover-full}, \dirmc in Section \ref{sec:dir-mc-full},
\multicut in Section \ref{sec:multicut-full} and Section \ref{sec:multicut2} and \kwaycut in Section \ref{sec:k-way-full}.
Section \ref{sec:conc} concludes the paper.

\paragraph{Acknowledgements.} We would like to thank Jakub Onufry Wojtaszczyk for
some early discussions on the kernelization of the graph separation problems.

%% file: preliminaries.tex
\section{Preliminaries}\label{sec:prelim}

\paragraph{Notation.} We use standard graph notation.
For a graph~$G$, by~$V(G)$ and~$E(G)$ we denote its vertex and edge set (or arc set in case
of directed graphs), respectively. For~$v \in V(G)$, its neighborhood~$N_G(v)$ is defined by~$N_G(v) = \{u: uv\in E(G)\}$, and~$N_G[v] = N_G(v) \cup \{v\}$ is the closed
neighborhood of~$v$.
We extend this notation to subsets of vertices:~$N_G[X] = \bigcup_{v \in X} N_G[v]$ and~$N_G(X) = N_G[X] \setminus X$.
For~$X\subseteq V(G)$ by~$\delta_G(X)$ we denote the set of edges in~$G$
with one endpoint in~$X$ and the other in~$V(G)\setminus X$.
For simplicity for a single vertex~$v$ we let~$\delta(v)=\delta(\{v\})$.
We omit the subscripts if no confusion is possible.
For a set~$X \subseteq V(G)$ by~$G[X]$ we denote the subgraph of~$G$ induced by~$X$.
For a set~$X$ of vertices or edges of~$G$, by~$G \setminus X$ we denote the graph
with the vertices or edges of~$X$ removed; in case of a vertex removal, we remove
also all its incident edges. For sets~$X,Y \subseteq V(G)$, the set~$E(X,Y)$
contains all edges of~$G$ that have one endpoint in~$X$ and the second endpoint in~$Y$.
In particular,~$E(X,X) = E(G[X])$ and $E(X, V(G)\setminus X) = \delta_G(X)$.
For a (directed) graph~$G$ by an~$st$-path we denote any path that starts in~$s$ and ends in~$t$.

For two disjoint vertex sets~$S$,~$T$ by an~$S$--$T$ cut we denote any set of edges, which removal ensures that there is no path from a vertex in~$S$ to a vertex in~$T$ in the considered graph. By minimum~$S$--$T$ cut we denote an~$S$--$T$ cut of minimum cardinality. 
\paragraph{Parameterized complexity.} In the parameterized complexity setting, an instance comes with an integer parameter~$k$ ---
formally, a parameterized problem~$Q$ is a subset of~$\Sigma^* \times \mathbb{N}$ for some finite alphabet~$\Sigma$.
We say that the problem is {\em{fixed parameter tractable}} ({\em{FPT}}) if there exists an algorithm solving any instance~$(x,k)$ in
time~$f(k) {\rm poly}(|x|)$ for some (usually exponential) computable function~$f$.
It is known that a problem is FPT iff
it is kernelizable: a kernelization algorithm for a problem~$Q$ takes an instance~$(x,k)$ and in
time polynomial in~$|x|+k$ produces an equivalent instance~$(x', k')$ (i.e.,~$(x,k)\in Q$ iff~$(x',k')\in Q$) such that~$|x'| + k' \leq g(k)$ for some computable function~$g$.
The function~$g$ is the {\em{size of the kernel}}, and if it is polynomial, we say that~$Q$ admits a polynomial kernel.

\paragraph{Kernelization lower bounds framework.}
We use the cross-composition technique introduced by Bodlaender et al.~\cite{cross-composition} which builds upon Bodlaender et al.~\cite{bodlaender:kernel} and Fortnow and Santhanam~\cite{fortnow:kernel}.
\begin{definition}[Polynomial equivalence relation \cite{cross-composition}]
An equivalence relation~$\mathcal{R}$ on~$\Sigma^\ast$ is called a {\em{polynomial equivalence
  relation}} if (1) there is an algorithm that given two strings~$x,y \in \Sigma^\ast$
  decides whether~$\mathcal{R}(x,y)$ in~$(|x|+|y|)^{O(1)}$ time; (2) for any finite set~$S \subseteq \Sigma^\ast$
  the equivalence relation~$\mathcal{R}$ partitions the elements of~$S$ into at most~$(\max_{x \in S} |x|)^{O(1)}$ classes.
\end{definition}

\begin{definition}[Cross-composition \cite{cross-composition}]
Let~$L \subseteq \Sigma^\ast$ and let~$Q \subseteq \Sigma^\ast \times \mathbb{N}$ be
a parameterized problem. We say that~$L$ {\em{cross-composes}} into~$Q$ if there is a polynomial
equivalence relation~$\mathcal{R}$ and an algorithm which, given~$t$ strings~$x_1, x_2, \ldots x_t$
belonging to the same equivalence class of~$\mathcal{R}$, computes an instance~$(x^\ast,k^\ast) \in \Sigma^\ast \times \mathbb{N}$ in time polynomial in~$\sum_{i=1}^t |x_i|$
such that (1)~$(x^\ast,k^\ast) \in Q$ iff~$x_i \in L$ for {\bf some}~$1 \leq i \leq t$; (2)~$k^\ast$ is bounded polynomially in~$\max_{i=1}^t |x_i| + \log t$.
\end{definition}

\begin{theorem}[\cite{cross-composition}, Theorem 9]\label{thm:cross-composition}
If~$L \subseteq \Sigma^\ast$ is NP-hard under Karp reductions
and~$L$ cross-composes into the parameterized problem~$Q$ that
has a polynomial kernel, then~$\unlesscompass$.
\end{theorem}

Behind Theorem \ref{thm:cross-composition} stands the following result of
Fortnow and Santhanam \cite{fortnow:kernel}.

\begin{definition}[\cite{bodlaender:kernel}]\label{def:distillation}
A {\em{distillation}} algorithm for a problem~$L \subseteq \Sigma^*$ into a set~$L' \subseteq \Sigma^*$ is a polynomial-time
algorithm that given~$t$ strings~$x_1,x_2, \ldots, x_t$ outputs a string~$y \in \Sigma^\ast$
such that (1)~$y \in L'$ iff~$x_i \in L$ for some~$1 \leq i \leq t$; (2)~$|y|$ is bounded
polynomially in~$\max_{i=1}^t |x_i|$.
\end{definition}

\begin{theorem}[\cite{fortnow:kernel}, Theorem 1.2]\label{thm:fortnow}
An NP-complete language does not admit a distillation algorithm into an arbitrary
set unless \unlesscompass{}.
\end{theorem}

By replacing the OR operation in Definition \ref{def:distillation} by the AND operation
we obtain the AND-conjecture.
\begin{conjecture}[AND-conjecture \cite{bodlaender:kernel}]
A coNP-complete language does not admit a distillation algorithm into itself.
\end{conjecture}
This conjecture motivates us to define the AND variant of a cross-composition algorithm.
\begin{definition}[AND-cross-composition]
Let~$L \subseteq \Sigma^\ast$ and let~$Q \subseteq \Sigma^\ast \times \mathbb{N}$ be
a parameterized problem. 
We say that~$L$ {\em{AND-cross-composes}} into~$Q$ if there is a polynomial
equivalence relation~$\mathcal{R}$ and an algorithm which, given~$t$ strings~$x_1, x_2, \ldots x_t$
belonging to the same equivalence class of~$\mathcal{R}$, computes an instance~$(x^\ast,k^\ast) \in \Sigma^\ast \times \mathbb{N}$ in time polynomial in~$\sum_{i=1}^t |x_i|$
such that (1)~$(x^\ast,k^\ast) \in Q$ iff~$x_i \in L$ for {\bf each}~$1 \leq i \leq t$; (2)~$k^\ast$ is bounded polynomially in~$\max_{i=1}^t |x_i| + \log t$.
\end{definition}

Following the lines of the proof of Theorem 9 in \cite{cross-composition} we obtain the following result
(for sake of completeness we include the formal proof).
\begin{theorem}\label{thm:and-cross-composition}
If~$L \subseteq \Sigma^\ast$ is NP-complete under Karp reductions
and~$L$ AND-cross-composes into a parameterized problem~$Q$, which 
has a polynomial kernel and whose unparameterized variant
(i.e., with the parameter appended to the instance in unary) is in NP,
  then the AND-conjecture fails.
\end{theorem}

\begin{proof}
In this proof we closely follow the lines of the proof of Theorem 9 of \cite{cross-composition}.

We show how the assumptions of the theorem lead to a distillation algorithm
for the coNP-complete language $\bar{L} = \Sigma^\ast \setminus L$.
Let $x_1,x_2,\ldots,x_t \in \Sigma^\ast$ and $m = \max_{i=1}^t |x_i|$.
First note that if $t > (|\Sigma|+1)^m$, then there are duplicates in the input instances
and we may remove them. Thus for the rest of the proof we assume that $t \leq (|\Sigma|+1)^m$,
in particular $\log t = O(m)$.

Using the polynomial equivalence relation $\prel$ assumed in the definition
of the AND-cross-composition, in polynomial time we partition the strings $x_i$
into $r$ equivalence classes $X_1,X_2,\ldots,X_r$. Note that $r$ is bounded polynomially
in $m$.

For each class $X_j$ we apply the assumed AND-cross-composition on the strings in
$X_j$, obtaining an instance $(y_j,k_j)$ of the parameterized problem $Q$.
We then apply the assumed kernelization algorithm to the instance $(y_j,k_j)$,
obtaining $(y_j',k_j')$. We note that $|y_j'| + k_j'$ is bounded polynomially
in $k_j$, which is bounded polynomially in $m$.
We infer that the total size of all instances $(y_j',k_j')$ for $1 \leq j \leq r$
is bounded polynomially in $m$.

As the unparameterized version of $Q$ is in NP, we may transform each
instance $(y_j',k_j')$ into a boolean formula $\phi_j$ of size polynomial in $m$
such that $\phi_j$ is satisfiable iff $(y_j',k_j') \in Q$.
Let $\phi = \bigwedge_{j=1}^r \phi_j$. Note that $|\phi|$ is bounded polynomially
in $m$ and $\phi$ is satisfiable if and only if $x_i \in L$ for all $1 \leq i \leq t$.
As $L$ is NP-complete, in polynomial time we may transform $\phi$ into an equivalent
instance $x$ of the language $L$, where $|x|$ is bounded polynomially in $m$.
We conclude by noting that we obtained a distillation algorithm for $\bar{L}$:
$x \in \bar{L}$ iff $x_i \notin L$ for some $1 \leq i \leq t$, that is,
$x_i \in \bar{L}$.
\end{proof}

Observe that any polynomial equivalence relation is defined on all
words over the alphabet~$\Sigma$ and for this reason
whenever we define a cross-composition,
we should also define how the relation behaves on words
that do not represent instances of the problem.
In all our constructions the defined relation puts all malformed instances
into one equivalence class, and the corresponding cross-composition
outputs a trivial NO-instance, given a sequence of malformed instances.
Thus, in the rest of this paper, we silently ignore the existence of malformed instances.

%% file: clique-cover-full.tex
\newcommand{\cC}{\mathcal{C}}

\section{Clique Cover}\label{sec:clique-cover-full}

\defproblemnoparam{\cliquecover}{An undirected graph $G$ and an integer $k$.}
{Does there exist a set of $k$ subgraphs of $G$, such that
each subgraph is a clique and each edge of $G$ is contained in at least one of these subgraphs?}

In this section we present both the cross-composition and the AND-cross-composition of \cliquecover{} parameterized by $k$.
We start with the AND-cross-composition since the construction we present is also used in the cross-composition.

\subsection{AND-cross-composition}

\begin{theorem}
\label{thm:clique-and-cross-full}
\cliquecover AND-cross-composes to \cliquecover parameterized by $k$.
\end{theorem}

\begin{proof}
For the equivalence relation $\mathcal{R}$ we take a relation
that puts two instances $(G_1,k_1)$, $(G_2,k_2)$ of \cliquecover are 
in the same equivalence class iff $k_1=k_2$ and
the number of vertices in $G_1$ is equal to the number of vertices in $G_2$.
Therefore, in the rest of the proof we assume that we are given a sequence $(G_i,k)_{i=0}^{t-1}$
of \cliquecover instances that are in the same equivalence class of $\mathcal{R}$
(to avoid confusion we number everything starting from zero in this proof).
Let $n$ be the number of vertices in each of the instances.
W.l.o.g. we assume that $n=2^{h_n}$ for a positive integer $h_n$, since otherwise we may add isolated vertices to each instance.
Moreover, we assume that $t=2^{h_t}$ for some positive integer $h_t$, since we may copy some instance if needed,
while increasing the number of instances at most two times.

Now we construct an instance $(G^\ast,k^\ast)$, where $k^\ast$ is polynomial in $n+k+h_t$.
Initially as $G^\ast$ we take a disjoint
union of graphs $G_i$ for $i=0,\ldots,t-1$ with added edges between every pair of vertices
from $G_a$ and $G_b$ for $a\not=b$.
Next, in order to cover all the edges between different instances with few cliques
we introduce the following construction.
Let us assume that the vertex set of $G_i$ is $V_i=\{v^i_0,\ldots,v^i_{n-1}\}$.
For each $0 \le a < n$, for each $0 \le b < n$ and for each $0 \le r < h_t$
we add to $G^\ast$ a vertex $w(a,b,r)$ which is adjacent to exactly one vertex in each $V_i$,
that is $v^i_j$ where $j=(a+b\lfloor\frac{i}{2^r}\rfloor) \bmod n$.
By $W$ we denote the set of all added vertices $w(a,b,r)$.
As the new parameter $k^\ast$ we set $k^\ast=|W|+k=n^2h_t+k$.
Note that $W$ is an independent set in $G^\ast$ and, moreover, each vertex in $W$ is non-isolated.

Let us assume that for each $i=0,\ldots,t-1$ the instance $(G_i,k)$ is a YES-instance.
To show that $(G^\ast,k^\ast)$ is a YES-instance we create a set $\cC$ of $k^\ast$ cliques.
We split all the edges of $G^\ast$ into the following groups:
(i) edges incident to vertices of $W$, (ii) edges between two different graphs $G_i$,$G_j$
and (iii) edges in each graph $G_i$.
For each vertex $w \in W$ we add to $\cC$ the subgraph $G^\ast[N[w]]$, which is a clique since
every two vertices from two different graphs $G_i,G_j$ are adjacent.
Moreover, let $\cC_i=\{C^i_0,\ldots,C^i_{k-1}\}$ be any solution for the instance $(G_i,k)$.
For each $\ell=0,\ldots,k-1$ we add to $\cC$ a clique $G^\ast\left[\bigcup_{i=0}^{t-1} C^i_\ell\right]$.
Clearly all the edges mentioned in (i) and (iii) are covered.
Consider any two vertices $v^i_x \in V_i$ and $v^j_y \in V_j$ for $i<j$.
Let $r$ be the greatest integer such that $(j-i)$ is divisible by $2^r$.
Note that $0 \le r < h_t$ and $z=\lfloor\frac{j}{2^r}\rfloor-\lfloor\frac{i}{2^r}\rfloor \equiv 1 \pmod 2$
since otherwise $(j-i)$ would be divisible by $2^{r+1}$.
Consequently, there exists $0 \le b < n$ satisfying the congruence $bz\equiv y-x \pmod n$,
since the greatest common divisor of $z$ and $n$ is equal to one (recall that $n$ is a power of $2$).
Therefore, when we set $a=y-b\lfloor\frac{j}{2^r}\rfloor$ we obtain
\begin{align*}
a+b\big\lfloor\frac{i}{2^r}\big\rfloor & \equiv b(\big\lfloor\frac{i}{2^r}\big\rfloor-\big\lfloor\frac{j}{2^r}\big\rfloor)+y\equiv y-bz\equiv x \pmod n \\
a+b\big\lfloor\frac{j}{2^r}\big\rfloor & \equiv y \pmod n
\end{align*}
and both $v^i_x,v^j_y$ belong to the clique of $\cC$ containing the vertex $w(a,b,r)$.

Now let us assume that $(G^\ast,k^\ast)$ is a YES-instance and let $\cC$ be a set of at most $k^\ast$ cliques in $G^\ast$ that cover every edge in $G^\ast$.
We define $\cC' \subseteq \cC$ as the set of these cliques in $\cC$ which contain at least two vertices
from some set $V_i$.
Since $W$ is an independent set in $G^\ast$, edges incident to two different vertices in $W$ need to be covered by two different cliques in $\cC$.
Moreover, no clique in $\cC'$ contains a vertex from $W$, because each vertex in $W$ is incident to
exactly one vertex in each $V_i$.
Therefore, $|\cC'| \le |\cC|-|W| \le k$ and a set $\cC_i=\{X\cap V_i : X \in \cC' \}$ for $i=0,\ldots,t-1$ is a solution for $(G_i,k)$, as
no clique in $\cC \setminus \cC'$ covers an edge between two vertices in $V_i$ for any $i=0,\ldots,t-1$.
Hence each instance $(G_i,k)$ is a YES-instance.
\end{proof}

As a consequence, by Theorem~\ref{thm:and-cross-composition} we obtain the following result.

\begin{corollary}
There is no polynomial kernel for the \cliquecover problem parameterized by $k$ unless the AND-conjecture fails.
\end{corollary}

\subsection{Cross-composition}

In this section we show cross-composition to \cliquecover, which we obtain
by extending the AND-cross-composition gadgets from the previous section.

\defproblemnoparam{\icliquecover}{An undirected graph $G$, an integer $k$
  and a set $\cC$ of $k+1$ cliques in $G$ covering all edges of $G$.}
{Does there exist a set of $k$ subgraphs of $G$, such that
each subgraph is a clique and each edge of $G$ is contained in at least one of the subgraphs?}

\begin{lemma}
\icliquecover is NP-complete with respect to Karp's reductions.
\end{lemma}

\begin{proof}
Clearly \icliquecover is in NP.

To prove that \icliquecover is NP-hard we show a reduction from  $3$-\colouring of $4$-regular planar graphs, which is NP-hard by~\cite{dailey}.
Let a $4$-regular planar graph $G=(V,E)$ be an instance of $3$-\colouring.
By Brooks theorem we know that $G$ is $4$-colourable (since by planarity, $G$ has no connected component isomorphic to $K_5$)
and we may find $4$-colouring of $G$ in polynomial time~\cite{brooks}.
Let $\bar{G}=(V,\bar{E})$ be the complement of $G$, that is an edge $e$ is in $\bar{E}$ iff $e$ does not belong to $E$.
To construct the graph $G'$ as the set of vertices we take two copies of $V$, namely $V_1 = \{v_1 : v\in V\}$, $V_2=\{v_2 : v \in V\}$.
For each edge $uv \in \bar{E}$ we add to $G'$ four vertices $w^{p,q}_{uv}$ for $1\leq p,q\leq 2$ and edges $w^{p,q}_{uv}u_p$, $w^{p,q}_{uv}v_q$, $u_pv_q$.
By $W$ we denote the set of all vertices $w^{p,q}_{uv}$ in $G'$.
Finally, for each $v \in V$ we add to $G'$ an edge $v_1v_2$ and set $k=|W|+3=4|\bar{E}|+3$.

In order to make $(G',k)$ a proper instance of \icliquecover we need also to construct a set $\cC$ of $k+1$ cliques covering all edges of $G'$.
Observe that to cover edges incident to vertices of $W$ we need at least $|W|$ cliques since $W$ is an independent set in $G'$.
Moreover, for each $w \in W$ the set $N_{G'}[w]$ is a clique in $G'$; hence w.l.o.g. any set of cliques covering all edges of $G'$
contains $|W|$ cliques of the form $N_{G'}[w]$ for $w \in W$ and those $|W|$ cliques cover all the edges of $G'$ except for $E'=\{v_1v_2 : v\in V\}$.
Note that to cover two different edges $u_1u_2,v_1v_2 \in E'$ we need $u_1$ and $v_2$ to be adjacent in $G'$, that is, non-adjacent in $G$.
Hence covering $E'$ with $l$ cliques is equivalent to colouring $G$ in $l$ colours. Since $G$ is $4$-colourable in an efficient way,
we can construct a set $\cC$ of $k+1$ cliques covering $G'$ obtaining an instance of \icliquecover, 
which is a YES-instance iff $G$ is $3$-colourable.
\end{proof}

Now the goal is to adjust the construction from the proof of Theorem \ref{thm:clique-and-cross-full} in order to obtain a classical cross-composition of \icliquecover into \cliquecover. Observe that we cannot easily relax the assumption that the clique cover of size $k+1$ is given in the input to just promising its existence, as the composition algorithm needs to be able to distinguish malformed instances from well-formed in the first place, which would not be the case unless $P=NP$. Moreover, the \icliquecover problem is trivially NP-hard with respect to Turing reductions; however, in order to make the composition work we need NP-completeness in Karp's sense.

\begin{theorem}
\label{thm:clique-cross-full}
\icliquecover cross-composes to \cliquecover parameterized by $k$.
\end{theorem}

\begin{proof}
We define the polynomial equivalence relation $\mathcal{R}$
in exactly the same way as in the proof of Theorem~\ref{thm:clique-and-cross-full},
that is we group instances according to their number of vertices and the
value of $k$.
Thus in the rest of the proof we assume we are given a sequence $(G_i,k,\cC_i)_{i=0}^{t-1}$
of \icliquecover instances that are in the same equivalence class of $\mathcal{R}$.
As in the proof of Theorem~\ref{thm:clique-and-cross-full}
we let $n$ be the number of vertices in each of the instances
and we assume $n=2^{h_n}$ and $t=2^{h_t}$.

Before we proceed to the proof let us give some intuition on what
follows. We would like to use the construction from Theorem~\ref{thm:clique-and-cross-full}
and extend it by adding exactly $h_t$ gadgets.
We show that any solution w.l.o.g. behaves in only one of two possible ways in every gadget.
Intuitively, each choice for the $j$-th gadget relaxes the constraint of using only $k$ cliques for half of the instances.
That is, choosing behaviour $b$, for $b=0,1$, allows using $k+1$ cliques, which are always sufficient by solution $\cC_i$ given as a part of the input,
for all instances with the $j$-th bit of the instance number equal to $b$.
Hence there is exactly one instance which is not relaxed by any of the $h_t$ gadgets,
so intuitively the gadgets may be viewed as an instance selector from $t$ instances.

\paragraph{Construction}
We create the instance of clique cover $(G^\ast,k^\ast)$ as in the proof of 
Theorem~\ref{thm:clique-and-cross-full}.
To obtain an instance $(G',k')$ we 
set $G'$ as $G^\ast$ and for each $j=1,\ldots,h_t$ we add to $G'$ a gadget $D_j$ containing exactly $6$ vertices
$V(D_j)=\{d^L_{j,1},d^L_{j,2},d^L_{j,3},d^R_{j,1},d^R_{j,2},d^R_{j,3}\}$ and $12$ edges $\binom{V(D_j)}{2} \setminus \{d^L_{j,r}d^R_{j,r} : 1 \le r \le 3\}$.
In other words, $D_j$ is a clique with a perfect matching removed (see Fig.~\ref{fig:edge-clique-cover-full}).
Let $V^L_j$ be the union of all sets $V_i=\{v^i_a : 0 \le a < n\}$ (recall that $V_i=V(G_i)$ is the set of vertices of the $i$-th instance) such that the $j$-th bit of the number $i$ written in binary is equal to zero,
whereas similarly $V^R_j$ is the set of vertices of all instances having the $j$-th bit of their number equal to one.
We make each vertex of $L_j=\{d^L_{j,r} : 1 \le r \le 3\}$ adjacent 
to each vertex of $V^L_j$ and we make each vertex of 
$R_j=\{d^L_{j,r} : 1 \le r \le 3\}$ adjacent 
to each vertex of $V^R_j$ in $G'$.
Finally, in order to allow easy coverage of the edges between $V(D_j)$ and $V^L_j \cup V^R_j$,
for each $0\le a < n$, $1 \le r \le 3$, and $Z \in \{L,R\}$ we add to $G'$ a vertex $s(a,r,Z)$ adjacent to each vertex in $\{v^i_a \in V^Z_j : 0 \le i < t\} \cup \{d^Z_{j,r}\}$.
Let $S$ be the set of all added vertices $s(a,r,Z)$.
As the parameter we set $k'=k^\ast + |S|+4h_t=n^2h_t+6nh_t+4h_t+k$.

\begin{figure}[htp]\label{fig:edge-clique-cover-full}
\begin{center}
\includegraphics[scale=1]{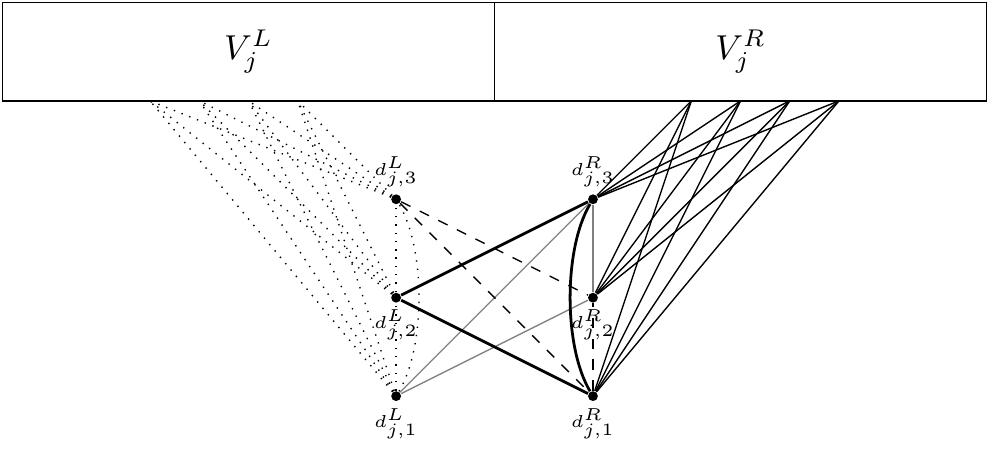}
\caption{The gadget $D_j$ set to relax the left part. Different styles of picturing the edges (thick, dotted, dashed or grayed) indicate, to which of the four cliques constructed for $D_j$ the edge belongs to.}
\end{center}
\end{figure}

\paragraph{Analysis}
We split all the edges of $G^\ast$ into the following groups:
\begin{itemize}
\item[(i)] edges incident to vertices of $W \cup S$ (recall that $W$ is a set defined as in the proof of Theorem~\ref{thm:clique-and-cross-full}),
\item[(ii)] edges between two different graphs $G_i$,$G_j$,
\item[(iii)] edges in each graph $G_i$,
\item[(iv)] edges within each gadget $D_j$.
\end{itemize}
First let us assume that for some $0 \le i_0 < t$ the instance $(G_i,k,\cC_i)$
of \icliquecover is a YES-instance.
We construct a set of cliques $\cC$.
For each vertex $x \in W \cup S$ we add to $\cC$ a clique $N_{G'}[x]$.
Hence by using $|W|+|S|=n^2h_t+6nh_t$ cliques we cover all edges of (i) and (ii) (by the same analysis as in the proof of Theorem~\ref{thm:clique-and-cross-full}).
Let us assume that each set $\cC_i$ is of the form $\cC_i = \{C^{i}_0,\ldots,C^{i}_{k}\}$
and also let us somewhat abuse the notation and assume that the $k$ cliques $C^{i_0}_0,\ldots,C^{i_0}_{k-1}$
form a solution for the YES-instances $(G_{i_0},k)$.
Let $(b_0b_1\ldots b_{h_t})_2$ be the binary representation of $i_0$.
Set $Z_j=L,Z_j'=R$ iff $b_j$ equals one and $Z_j=R,Z_j'=L$ otherwise, for $j=0,\ldots,h_t-1$.
For each $j=0,\ldots,h_t-1$ 
we add to $\cC$ exactly $4$ cliques
$\left\{d^{Z_j}_{j,1},d^{Z_j'}_{j,2},d^{Z_j'}_{j,3}\right\}$,
$\left\{d^{Z_j}_{j,2},d^{Z_j'}_{j,1},d^{Z_j'}_{j,3}\right\}$,
$\left\{d^{Z_j}_{j,3},d^{Z_j'}_{j,1},d^{Z_j'}_{j,2}\right\}$,
$\left\{d^{Z_j}_{j,1},d^{Z_j}_{j,2},d^{Z_j}_{j,3}\right\}\cup \left(\bigcup_{i=0}^{t-1} \left(C^i_k \cap V_j^{Z_j}\right)\right)$ (see Fig.~\ref{fig:edge-clique-cover-full}).
It is easy to verify that the $4$ added sets are indeed cliques in $G'$
and that they cover edges of (iv).
Note that the last of the four cliques contains the last clique of the solution $\cC_i$
for each instance $i$ that has $j$-th bit equal to $b_j$.
Consequently, some of the edges of (iii) are covered.
We add exactly $k$ more cliques to $\cC$, that is for each $\ell=0,\ldots,k-1$ we add to $\cC$ a clique $\bigcup_{i=0}^{t-1} C^i_\ell$.
Since for each $i=0,\ldots,t-1$ such that $i\not= i_0$ there is a clique in $\cC$ containing $C^i_k$
and the first $k$ cliques of $\cC_{i_0}$ form a cover of $G_i$, we infer
that all edges of (iii) are covered and, therefore, $(G',k')$ is a YES-instances of \cliquecover.

In the other direction, assume that $(G',k')$ is a YES-instance of \cliquecover
and let $\cC$ be any solution containing $k'$ cliques.
Since $W \cup S$ is an independent set in $G'$ and each vertex in $W \cup S$ is not isolated,
we infer that there are at least $|W|+|S|$ cliques in $\cC$ containing a vertex of $W \cup S$.
Let $\cC' \subseteq \cC$ be the set of at most $k'-|W|-|S|$ cliques of $\cC$ which have empty intersection with $W \cup S$.
Each vertex in $W\cup S$ is adjacent to exactly one vertex in each $V_i$ for $0\le i < t$ and at most one
vertex in $V(D_j)$ for $0 \le j < h_t$, therefore cliques of $\cC'$ cover all edges of (iii) and (iv).
Moreover, $|\cC'| \le k'-|W|-|S|=4h_t+k$.
We use the following lemma which we prove afterwards.
\begin{lemma}
\label{lem:clique-nice-form-full}
One can modify the set $\cC'$ maintaining coverage of edges of (iii) and (iv) and not incrementing its size, while
at the same time for each $j=0,\ldots,h_t-1$ obeying the following conditions:
\begin{itemize}
  \item[(a)] $\cC'$ contains exactly $3$ cliques containing both a vertex of $L_j$ and $R_j$ (recall that $L_j \cup R_j = V(D_j)$),
  \item[(b)] $\cC'$ contains exactly $1$ clique $C$ having exactly one of the two intersections $C \cap L_j$,$C \cap R_j$ non-empty.
\end{itemize}
\end{lemma}
Before we prove Lemma~\ref{lem:clique-nice-form-full} let us finish the proof of Theorem~\ref{thm:clique-cross-full}
assuming that Lemma~\ref{lem:clique-nice-form-full} holds.
Since no two vertices from different gadgets $D_{j_1},D_{j_2}$ are adjacent, we infer that $\cC'$
contains exactly $4h_t$ cliques containing a vertex from some $D_j$ for $0 \le j < h_t$.
For each $j=0,\ldots,h_t-1$, let $C$ be the clique from (b) of Lemma~\ref{lem:clique-nice-form-full}.
If $C \cap L_j\not=\emptyset$, we take $I_j \subseteq \{0,\ldots,t-1\}$ to be the set of instance numbers that have the $j$-th bit
equal to one, whereas if $C \cap R_j \not=\emptyset$, then as $I_j$ we take all
the instance numbers that have the $j$-th bit equal to zero.
Observe that $\bigcap_{j=0}^{h_t-1}I_j$ contains exactly one element and denote it by $i_0$.
By Lemma~\ref{lem:clique-nice-form-full} every clique from $\cC'$, that contains a vertex of $V(D_j)$ for any $j=0,\ldots,h_t-1$, has to be disjoint with $V_{i_0}$. Indeed, cliques containing vertices from $V(D_j)$ satisfy (a) or (b) from Lemma~\ref{lem:clique-nice-form-full};
a clique from (a) contains both vertices of $L_j$ and $R_j$ and no vertex of $V_{i_0}$ is incident to both $L_j$ and $R_j$,
while a clique from (b) contains vertices of the one of the sets $L_j,R_j$ which is not connected to anything in $V_{i_0}$.
Therefore, edges of $G_{i_0}$ are covered by $|\cC'|-4h_t \le k$ cliques, being intersections of the remaining cliques in $\cC'$ with $V_{i_0}$. Consequently, $(G_{i_0},k)$ is a YES-instance, which finishes the proof of Theorem~\ref{thm:clique-cross-full}.
\end{proof}

\begin{proof}[Proof of Lemma~\ref{lem:clique-nice-form-full}]
Let $j$ be any index for which the lemma does not hold, i.e., the cliques containing vertices from $V(D_j)$ do not behave as in the lemma statement. We refine the set $\cC'$ repairing its behaviour on gadget $D_j$ and not spoiling the behaviour on other gadgets. By applying this reasoning to all the gadgets that need repairing, we prove the lemma.

Let us denote by $\Hh_j$ the set of cliques from $\cC'$ that contain a vertex from $V(D_j)$, while let $\Ii_j\subseteq \Hh_j$ be the set of these cliques from $\Hh_j$, which have nonempty intersection with both $L_j$ and $R_j$. The goal is to obtain a situation, when $|\Hh_j|=4$ and $|\Ii_j|=3$ for every $j$. During refining the set $\cC'$ we will change only the set $\Hh_j$. As there are no edges between the gadgets, the sets $\Hh_j$ are always disjoint, so our repairs do not spoil the behaviour of $\cC'$ on other gadgets.

Let $C \in \Hh_j$. Observe that $|C \cap (L_j \cup R_j)| \le 3$, since $C$ contains at most one of the vertices $D^L_{j,p},D^R_{j,p}$ for $p=1,2,3$.
Therefore, each clique in $\Hh_j$ covers at most $3$ out of $12$ edges of $D_j$ and, consequently, $|\Hh_j|\geq 4$.
Consider two cases.

First assume that $\Hh_j\setminus \Ii_j\neq \emptyset$, i.e., there exists a clique $C_0 \in \cC'$ which has an element of $L_j \cup R_j$,
but has exactly one of the two intersections $C_0 \cap L_j$, $C_0 \cap R_j$ non-empty.
By symmetry assume that $(C_0 \cap (L_j \cup R_j)) \subseteq L_j$.
Note that $C_0 \cup L_j$ also forms a clique, hence w.l.o.g. we may assume that $L_j \subseteq C_0$.
We know that $|\Ii_j|\geq 3$, since $\cC'$ covers all the $6$ edges of $E(L_j,R_j)$,
whereas each clique in $\cC'$ covers at most two of them.
Note that each clique from $\Ii_j$ is entirely contained in $D_j$,
since there is no vertex outside of $D_j$ which is adjacent to both a vertex of $L_j$ and a vertex of $R_j$.
Therefore, we may substitute the whole $\Ii_j$ with just $3$ cliques:
\begin{align*}
C_1:= \{d^L_{j,1},d^R_{j,2},d^R_{j,3}\}\,,\\
C_2:= \{d^L_{j,2},d^R_{j,1},d^R_{j,3}\}\,,\\
C_3:= \{d^L_{j,3},d^R_{j,1},d^R_{j,2}\}\,,
\end{align*}
maintaining the property that $\cC'$ covers all the edges. Observe that cliques $C_0,C_1,C_2,C_3$ already cover all the edges in $D_j$.
Finally, after this modification for any $C \in \cC'$ such that $C \not\in \{C_0,C_1,C_2,C_3\}$
we set $C := C \setminus (L_j \cup R_j)$ and satisfy both constraints (a) and (b) of the lemma for this particular $j$.

Now assume that for $\Ii_j=\Hh_j$.
Similarly as in the previous case, for each clique $C \in \Ii_j$ we have $C \subseteq (L_j \cup R_j)$. As $|\Ii_j|=|\Hh_j|\geq 4$, we may substitute the whole $\Ii_j=\Hh_j$ with just $4$ cliques:
\begin{align*}
C_0:= \{d^L_{j,1},d^L_{j,2},d^L_{j,3}\}\,,\\
C_1:= \{d^L_{j,1},d^R_{j,2},d^R_{j,3}\}\,,\\
C_2:= \{d^L_{j,2},d^R_{j,1},d^R_{j,3}\}\,,\\
C_3:= \{d^L_{j,3},d^R_{j,1},d^R_{j,2}\},
\end{align*}
out of which exactly one is not contained in the new $\Ii_j$. As these cliques cover all the edges of $D_j$ and every removed clique was entirely contained in $D_j$, all the edges of $G'$ are still covered.
In this way we make the modified set $\cC'$ satisfy both constraints (a) and (b) for the considered value of $j$.
\end{proof}

\begin{corollary}
There is no polynomial kernel for the \cliquecover problem parameterized by $k$ unless \unlesscompass{}.
\end{corollary}

%% file: directed-multiway-full.tex
\section{Directed Multiway Cut}\label{sec:dir-mc-full}

In the \dirmc problem we want to disconnect
every pair of terminals in a directed graph.
The problem was previously studied in 
the following two versions.

\defproblemnoparam{\diremc}{A directed graph $G=(V,A)$, a set
 of terminals $T\subseteq V$ and an integer $p$.}
{Does there exist a set $S$ of at most $p$ arcs in $A$,
such that in $G \setminus S$ there is no path
between any pair of terminals in $T$?}

\defproblemnoparam{\dirvmc}{A directed graph $G=(V,A)$, a set
 of terminals $T\subseteq V$, a set of forbidden
  vertices $V^\infty \supseteq T$ and an integer $p$.}
{Does there exist a set $S$ of at most $p$ vertices in $V \setminus V^\infty$,
such that in $G \setminus S$ there is no path between any pair of terminals in $T$?}

As a side note, observe that by replacing each vertex of $V^\infty\setminus T$ with a $p+1$-clique (i.e., a graph on $p+1$ vertices pairwise connected by arcs in both directions), one can reduce
the above \dirvmc version to a version, where the solution is allowed to remove any nonterminal vertex.
Moreover, it is well known, that given an instance $I$ of one of the two problems above,
one can in polynomial time create an equivalent instance $I'$ of the other problem,
where both the number of terminals and the value of $p$ remain unchanged (e.g. see~\cite{marx:dir-mc}).
Therefore we show cross-composition to \dirvmc and as a corollary we
prove that \diremc also does not admit a polynomial kernel.
The starting point is the following restricted variant of \dirvmc, which we prove to be NP-complete
with respect to Karp reductions.

\defproblemnoparam{\prepdirvmc}{A directed graph $G=(V,A)$, 
  two terminals $T=\{s_1,s_2\}$, a set of forbidden
  vertices $V^\infty \supseteq T$ and an integer $p$.
  Moreover, after removing any set of at most $p/2$ vertices of $V \setminus V^\infty$,
  both an $s_1s_2$-path and an $s_2s_1$-path remain.
}
{Does there exist a set $S$ of at most $p$ vertices in $V \setminus V^\infty$,
such that in $G \setminus S$ there is no $s_1s_2$-path nor $s_2s_1$-path?}

The assumption that any set of size at most $p/2$ can not hit all the paths from $s_1$ to $s_2$ (and similarly from $s_2$ to $s_1$)
will help us in constructing cross-composition.

\begin{lemma}
\prepdirvmc is NP-complete with respect to Karp's reductions.
\end{lemma}

\begin{proof}
Note that in order to show that the problem is in NP we need to argue that we can verify the condition
concerning removal of $p/2$ vertices.
However, this can be checked by a polynomial-time algorithm computing min $s_1$--$s_2$ cut and min $s_2$--$s_1$ cut.
If any of those cuts is of size at most $p/2$, then the instance is not a proper instance of \prepdirvmc.

To prove that the problem is NP-hard we use the NP-completeness result of Garg et al.~\cite{mc-apx-2}
for \dirvmc with two terminals.
Consider an instance $I=(G,T=\{s_1,s_2\},V^\infty,p)$ of \dirvmc.
As the graph $G'$ we take $G$ with $z=p+1$ vertices $\{u_1,\ldots,u_z\}$
added.
In $G'$ for $i=1,\ldots,z$ we add the following four arcs $\{(s_1,u_i),(u_i,s_1),(u_i,s_2),(s_2,u_i)\}$.
Let $I'=(G',T,V^\infty,p+z)$ be an instance of \prepdirvmc.
Since after removal of less than $z$ vertices in $G'$ at least one vertex $u_i$ remains,
we infer that $I'$ is indeed a \prepdirvmc instance.
To prove that $I$ is a YES-instance iff $I'$ is a YES-instance it is enough to observe
that any solution in $I'$ contains all the vertices $\{u_1,\ldots,u_z\}$.
\end{proof}

Equipped with the \prepdirvmc problem definition, we are ready to show a cross-composition into 
\dirvmc parameterized by $p$.

\begin{theorem}
\prepdirvmc cross-composes into \dirvmc with two terminals, parameterized by the size of the cutset $p$.
\end{theorem}

\begin{proof}
For the equivalence relation $\mathcal{R}$, we take a relation
that groups the input instances according to the value of $p$.
Formally $(G_i,T_i,V^\infty_i,p_i)$ and $(G_j,T_j,V^\infty_j,p_j)$
are in the same equivalence class in $\mathcal{R}$ iff $p_i=p_j$.
Therefore, we assume that we are given 
a sequence $I_{i} = (G_i,T_i=\{s_1^i,s_2^i\},V^\infty_i,p)_{i=1}^{t}$ of \prepdirvmc instances that 
are in the same equivalence class of $\mathcal{R}$.

As the graph $G'$ we take disjoint union of all the graphs $G_i$.
Moreover for each $i=1,\ldots,t-1$, in $G'$ we identify the vertices $s_2^i$ and $s_1^{i+1}$.
Let $I'=(G',\{s_1^1,s_2^t\},\bigcup_{i=1}^t V^\infty_i,p)$ be an instance of \dirvmc.
Note that $\bigcup_{i=1}^t V^\infty_i$ contains both terminals from all input instances.

Let us assume that there exists $1 \le i_0 \le t$ such that $I_{i_0}$ is
a YES-instance of \prepdirvmc, and let $S \subseteq V(G_i) \setminus V_i^\infty$ be any solution for $I_{i_0}$.
Since any $s_1^1s_2^t$-path and any $s_2^ts_1^1$-path in $G'$ goes through both $s_1^{i_0}$ and $s_2^{i_0}$,
we observe that $G'\setminus S$ is a solution for $I'$ and, consequently, $I'$ is a YES-instance.

In the other direction, let us assume that $I'$ is a YES-instance.
Let $S \subseteq V(G) \setminus \bigcup_{i=1}^t V^\infty_i$
by any solution for $I'$.
Observe that if the set $S$ contains at most $p/2$ vertices of $V(G_i) \setminus V_i^\infty$
for some $1 \le i \le t$, then $S \setminus V(G_i)$ is also a solution for $I'$, since
after removing at most $p/2$ vertices of $V(G_i)$ there is still a path both from $s_1^i$ to $s_2^i$ and
from $s_2^i$ to $s_1^i$.
Because $|S| \le p$, we infer that w.l.o.g. $S$ contains only vertices of a single set $V(G_{i_0})$
for some $1 \le i_0 \le t$. Therefore, $I_{i_0}$ is a YES-instance.
\end{proof}

The equivalence of \dirvmc and \diremc together with Theorem~\ref{thm:cross-composition} give us the following corollary.

\begin{corollary}
Both \dirvmc and \diremc do not admit a polynomial kernel when parameterized by $p$ unless \unlesscompass, even in the case of two terminals.
\end{corollary}

%% file: multicut-full.tex
\newcommand{\cT}{\mathcal{T}}

\section{Multicut}\label{sec:multicut-full}

In this section we prove that both the
edge and vertex versions of the \multicut
problem do not admit a polynomial kernel,
when parameterized by the size of the cutset.

\defproblemnoparam{{\sc Edge} ({\sc Vertex}) \multicut}{An undirected graph~$G=(V,E)$, a set
of pairs of terminals~$\cT=\{(s_1,t_1),\ldots,(s_k,t_k)\}$ and an integer~$p$.}
{Does there exists a set~$S \subseteq E$ ($S \subseteq V)$ such that no connected component
of~$G\setminus S$ contains both vertices~$s_i$ and~$t_i$, for some~$1 \le i \le k$?}

It is known that the vertex version of the \multicut problem is at least as hard
as the edge version. 

\begin{lemma}[folklore]
There is a polynomial time algorithm, which given an instance~$I=(G,\cT,p)$ of \emulticut
produces an instance~$I'=(G',\cT',p)$ of \vmulticut, such that~$I$ is a YES-instance
iff~$I'$ is a YES-instance.
\end{lemma}

In order to show a cross-composition into the \multicut problem parameterized by~$p$
we consider the following restricted variant of the \mc problem with three terminals.

\defproblemnoparam{\mc}{An undirected graph~$G=(V,E)$, a set
 of three terminals~$T=\{s_1,s_2,s_3\} \subseteq V$ and an integer~$p$.
}
{Does there exist a set~$S$ of at most~$p$ edges in~$E$,
such that in~$G \setminus S$ there is no path between any pair of terminals in~$T$?}

\defproblemnoparam{\prepmc}{An undirected graph~$G=(V,E)$, a set
 of three terminals~$T=\{s_1,s_2,s_3\} \subseteq V$ and an integer~$p$.
 An instance satisfies: (i)~$\deg(s_1)=\deg(s_2)=\deg(s_3)=d>0$,
  (ii) for each~$j=1,2,3$ and any non-empty set~$X \subseteq V\setminus T$ we have~$|\delta(X \cup \{s_j\})| > d$, and (iii)~$d \le p < 2d$.
}
{Does there exist a set~$S$ of at most~$p$ edges in~$E$,
such that in~$G \setminus S$ there is no path between any pair of terminals in~$T$?}

Condition (i) ensures that degrees of all the terminals are equal,
whereas condition (ii) guarantees that the set of edges incident to a terminal~$s_j$
is the only minimum size~$s_j$--$(T\setminus \{s_j\})$ cut.
Having both (i) and (ii), condition (iii) verifies whether an instance is not 
a trivially YES- or NO-instance, because by (i) and (ii) there is no solution of size less than~$d$ 
and removing all the edges incident to two terminals always gives a solution of size at most~$2d$.


\begin{lemma}
\prepmc is NP-complete with respect to Karp's reductions.
\end{lemma}

\begin{proof}
To prove the lemma we may
observe that the first NP-hardness reduction to the \mc problem
by Dahlhaus et al. \cite{nmc-np-hardness} in fact yields
a \prepmc instance. For sake of completeness, we 
present here how to reduce an arbitrary instance of
the \mc problem with three terminals to a \prepmc instance.

Let~$I=(G,T=\{s_1,s_2,s_3\},p)$ be an instance of \mc.
As observed by Marx~\cite{marx:cut}, we can assume that for each terminal~$s_i$ the cut~$\delta(s_i)$
is the only minimum cardinality~$s_i$--$(T \setminus \{s_i\})$ cut, since otherwise
w.l.o.g.\ we may contract some edge incident to~$s_i$ obtaining a smaller equivalent instance.
Therefore condition (ii) would be satisfied if only degrees of terminals were equal.
Let~$G_1$,~$G_2$,~$G_3$ be three copies of the graph~$G$, where terminals in the~$i$-th copy
are denoted by~$T_i=\{s^i_1,s^i_2,s^i_3\}$.
Construct a graph~$G'$ as a disjoint union of~$G_1$,~$G_2$ and~$G_3$.
Next in~$G'$ we identify vertices~$\{s^1_1,s^2_2,s^3_3\}$ into a single vertex~$s_1'$,
similarly identify vertices~$\{s^1_2,s^2_3,s^3_1\}$ into a single vertex~$s_2'$,
and finally identify vertices~$\{s^1_3,s^2_1,s^3_1\}$ into a single vertex~$s_3'$.
Let~$I'=(G',T'=\{s_1',s_2',s_3'\},p'=3p\}$.
Observe that due to the performed identification~$I'$ is a YES-instance of \mc iff~$I$ is a YES-instance of \mc.
Therefore, to finish the reduction it suffices to argue that~$I'$ satisfies (i), (ii) and (iii).

Let~$d=\sum_{i=1}^3 \deg_G(s_i)$.
Note that in~$G'$ for each~$i=1,2,3$ we have~$\deg_{G'}(s_i')=d$, hence
condition (i) is satisfied.
Observe that if there exists~$1 \le j \le 3$ and~$s_j'$--$(T'\setminus \{s_j'\})$ cut in~$G'$
of size at most~$d$, which is different from~$\delta(s_j')$, then
there exists~$1 \le r \le 3$ and~$s_r$--$(T\setminus \{s_r\})$ cut in~$G$ of size at most~$\deg_G(s_r)$
which is different from~$\delta_G(s_r)$, a contradiction.
Hence condition (ii) is satisfied.
Unfortunately, it is possible that~$p' \le d$ or~$p' \ge 2d$.
However, if~$p' \ge 2d$ then
clearly~$I'$ is a YES-instance (we can remove edges incident to two terminals), and hence~$I$ is a YES-instance.
On the other hand if~$p' < d$, then~$I'$ (and consequently~$I$) is a NO-instance, 
since any~$s_1'$--$\{s_2',s_3'\}$ cut has size at least~$d$.
Therefore, if condition (iii) is not satisfied, then in polynomial time
we can compute the answer for the instance~$I$, and as the instance~$I'$ we set a trivial YES- or NO-instance.
\end{proof}

\begin{theorem}
\label{thm:cross-multicut-full}
\prepmc cross-composes into \emulticut parameterized the size of the cutset~$p$.
\end{theorem}

\begin{proof}
For the equivalence relation~$\mathcal{R}$, we take a relation
where all well-formed instances are grouped according to the values of~$p$ and~$d$.
Formally,~$(G_i,T_i,p_i)$ and~$(G_j,T_j,p_j)$
are in the same equivalence class in~$\mathcal{R}$ iff~$p_i=p_j$ 
and the degree of each terminal in~$G_i$ equals
the degree of each terminal in~$G_j$.
Therefore, we assume that we are given 
a sequence~$I_{i} = (G_i,T_i=\{s_1^i,s_2^i,s_3^i\},p)_{i=0}^{t-1}$ of \prepmc instances that 
are in the same equivalence class of~$\mathcal{R}$ (note that we number instances starting from~$0$).
Let~$d$ be the degree of each terminal in each of the instances.
W.l.o.g. we assume that~$t\ge 5$ is an odd integer, since we may copy some instances if needed, and let~$h=(t-1)/2$.

\paragraph{Construction} 
Let~$G'$ be the disjoint union of all graphs~$G_i$ for~$i=0,\ldots,t-1$.
For each~$i=0,\ldots,t-1$ we add~$d$ parallel edges between vertices~$s^i_2$ and~$s^{(i+1) \bmod t}_1$.
To the set~$\cT$ we add exactly~$t$ pairs, that is for each~$i=0,\ldots,t-1$ we add
to~$\cT$ the pair~$(s_i'=s_3^i,t_i'=s_3^{(i+h)\bmod t})$.
We set~$p'=p+d$ and~$I'=(G',\cT,p')$ is the constructed \emulticut instance.
Note that in order to avoid using parallel edges it is enough to subdivide them.

\begin{figure}[htp]
\begin{center}
\includegraphics[scale=1]{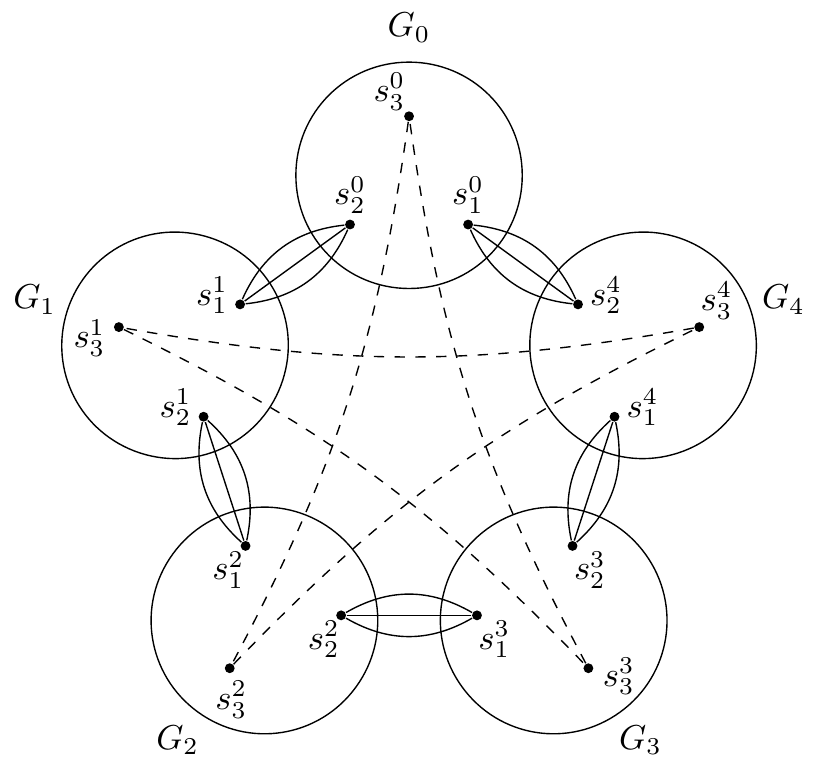}
\caption{Construction of the graph~$G'$ for~$t=5$ and~$d=3$. Dashed edges represent pairs of vertices in~$\cT$.}
\label{fig:multicut-full}
\end{center}
\end{figure}

\paragraph{Analysis}
First assume that there exists~$0 \le i_0 < t$ such that~$I_{i_0}$ is a YES-instance of \prepmc
and let~$S \subseteq E(G_{i_0})$ be any solution for~$I_{i_0}$.
Let~$S_1$ be the set of edges in~$G'$ between~$s_2^{(i_0+h) \bmod t}$ and~$s_1^{(i_0+h+1) \bmod t}$ 
(see Fig.~\ref{fig:multicut-full}).
We prove that~$S'=S \cup S_1$ is a solution for~$I'$.
Observe that~$|S'|=|S|+|S_1| \le p+d$. Consider any pair~$(s',t')\in \cT$ such that~$s',t' \not= s^{i_0}_3$.
Note that in~$G' \setminus S'$ there is neither an~$s^{i_0}_1s^{i_0}_2$-path, nor
an~$s^{(i_0+h) \bmod t}_2s^{(i_0+h+1)\bmod t}_1$-path.
Therefore, there is no~$s't'$-path~$G' \setminus S'$.
Moreover, in~$G' \setminus S'$ there is neither an~$s^{i_0}_3s^{i_0}_1$-path,
nor an~$s^{i_0}_3s^{i_0}_2$-path.
Consequently, for each~$(s',t') \in \cT$, where~$s'=s^{i_0}_3$ or~$t'=s^{i_0}_3$,
there is no~$s't'$-path in~$G' \setminus S'$, so~$I'$ is a YES-instance of \emulticut.

Now assume that~$I'$ is a YES-instance and our goal 
is to show that for some~$0 \le i < t$ the instance~$I_i$ is a YES-instance.
Let~$E_i=E(G_i)$ and let~$E_i'$ be the set of edges between~$s^{i}_2$ and~$s^{(i+1) \bmod t}_1$ in~$G'$.
Let~$S' \subseteq E(G')$ be any solution for~$I'$.
Note that if for some~$E_i'$, where~$0 \le i < t$, the set~$S'$ contains less than~$d$ edges from the set~$E_i'$, then~$S' \setminus E_i'$
is also a solution for~$I'$.
By conditions (i) and (ii) of the \prepmc problem definition we have the following: if for some~$0 \le i < t$ the set~$S'$ contains less than~$d$ edges from the set~$E_i$, then~$S' \setminus E_i$
is also a solution for~$I'$.
Indeed, if~$S'$ contains less than~$d$ edges from~$E_i$, then in the graph~$G' \setminus E_i$
all the vertices~$s^i_1$,~$s^i_2$,~$s^i_3$ are in the same connected component,
since otherwise for some~$a \in T_i$ there would be an~$a$--$(T_i \setminus \{a\})$ cut in~$G_i$ of size smaller than~$d$.
Let us recall that~$|S'| \le p'=p+d < 3d$.
Therefore, w.l.o.g. we may assume that the set~$S'$ has non-empty intersection with at most two sets from the set~$\mathcal{E}=\{E_0,\ldots,E_{t-1},E_0',\ldots,E_{i-1}'\}$.
Moreover we assume that if~$S'$ has non-empty intersection with some set from~$\mathcal{E}$, then this intersection is of size at least~$d$.

\textbf{Case 1.} Consider the case, when~$S'$ has an empty intersection with each of the sets~$E_i$ for~$0 \le i < t$.
Since~$|E_i'| = d$ and~$p' \ge 2d$, w.l.o.g.~$S'$ has a non-empty intersection with 
exactly two sets~$E_{i_0}'$,~$E_{i_1}'$ for~$i_0 \not= i_1$.
Since~$t$ is odd, in the graph~$G'\setminus S'$ either there is an~$s^{i_0}_3s^{(i_0-h) \pmod t}_3$-path or
an~$s^{(i_0+1)\bmod t}_3s^{(i_0+1+h)\bmod t}_3$-path.
Hence a contradiction.

\textbf{Case 2.}
Next assume that~$S'$ has a non-empty intersection with some set~$E_{i_0}$ for~$0 \le i_0 < t$.
By symmetry w.l.o.g. we may assume that~$i_0=0$.
Since the set~$S'$ hits all the~$s^{1}_3s^{h+1}_3$-paths as well as all the~$s^{h}_3s^{t-1}_3$-paths
in the graph~$G'$, we infer that~$S'$ has non-empty intersection with exactly one of the sets~$E_h$,~$E_h'$,~$E_{h+1}$.

\textbf{Case 2.1.}
In this case we assume that~$S'$ has a non-empty intersection with~$E_h'$.
Since~$S'$ hits all~$s^0_3s^h_3$-paths in~$G'$, in~$G_0 \setminus S'$ there is no~$s^0_3s^0_2$-path.
Similarly, since~$S'$ hits all~$s^0_3s^{h+1}_3$-paths in~$G'$, in~$G_0 \setminus S'$ there is no~$s^0_3s^0_1$-path.
Moreover, since~$S'$ hits all~$s^{t-1}_3s^{h-1}_3$-paths in~$G'$, in~$G_0 \setminus S'$ there is no~$s^0_1s^0_2$-path.
Since~$|S'| \le p' = p+d$ and~$|S' \cap E_h'| = d$, we infer that~$|S' \cap E_0| \le p$, and,
consequently,~$I_0$ is a YES-instance.

\textbf{Case 2.2.}
Since~$S'$ has a non-empty intersection with one of the sets~$E_h$,~$E_{h+1}$, by symmetry
we assume that~$S' \cap E_h \not=\emptyset$.
Recall that~$t \ge 5$, and hence~$h > 1$. 
Since~$S'$ hits all~$s^1_3s^{h+1}_3$-paths and all~$s^1_3s^{t+1-h}_3$-paths in~$G'$,
we infer that in~$G_0 \setminus S'$ there is no~$s^0_1s^0_2$-path and
in~$G_h \setminus S'$ there is no~$s^h_1s^h_2$-path.
Moreover,~$S'$ hits all~$s^0_3s^{h+1}_3$-paths and all~$s^h_3s^{t-1}_3$-paths in~$G'$;
therefore, in~$G_0 \setminus S'$ there is no~$s^0_3s^0_1$-path and
in~$G_h \setminus S'$ there is no~$s^h_3s^h_2$-path.
Finally since~$S'$ hits all~$s^0_3s^h_3$-paths in~$G'$, either in~$G_0 \setminus S'$
there is no~$s^0_3s^0_2$-path, or in~$G_h\setminus S'$ there is no~$s^h_3s^h_1$-path.
Since~$|S'| \le p+d$,~$|S' \cap E_0| \ge d$ and~$|S' \cap E_h| \ge d$,
we infer that~$|S' \cap E_0| \le p$ and~$|S' \cap E_h| \le p$.
Consequently, either~$I_0$ or~$I_h$ is a YES-instance, which finishes the proof of Theorem~\ref{thm:cross-multicut-full}.
\end{proof}

%% file: multicut2.tex
\section{Alternative cross-composition of Multicut}\label{sec:multicut2}

In this section we present an alternative proof of cross-composition
to \emulticut parameterized by the size of the cutset.
Despite the fact that the cross-composition presented in this section
is more involved comparing to the one presented in Section~\ref{sec:multicut-full},
we find the ideas used here more general, which might be helpful in designing future cross-compositions for other problems.

\begin{theorem}
\label{thm:cross-multicut2}
\prepmc cross-composes into \emulticut parameterized the size of the cutset~$p$.
\end{theorem}

\begin{proof}
We start by defining a relation~$\prel$ on \prepmc instances,
which groups instances according to the size of the cutset~$p$.
Formally~$(G,T,p)$ is in relation~$\prel$ with~$(G',T',p')$ iff~$p=p'$.
Clearly,~$\prel$ is a polynomial equivalence relation.
Hence we assume that we are given~$t \ge 1$ instances~$I_i=(G_i,T_i=\{s^i_1,s^i_2,s^i_3\},p)$,
for~$1 \leq i \leq t$, of the \prepmc problem (note that we number instances starting from~$1$).

\paragraph{Construction.}
Let~$M=p+1$ and~$M_{\infty}=6M+p+1$. In our construction
each edge of \emulticut instance will have one of three possible weights~$\{1,M,M_\infty\}$.
We can implement those weights by putting~$1$,~$M$ or~$M_{\infty}$ parallel
edges and subdividing them to obtain a simple graph (note that both~$M$ and~$M_\infty$ are polynomially bounded in~$p$).
Initially as the graph~$G'$ we take a disjoint union of two cycles~$C_1$,~$C_2$, each containing exactly~$3(t+1)$ vertices.
To simplify presentation, each vertex on each of the two cycles 
has two different names, that is~$C_j=(x^j_0,x^j_1,\ldots,x^j_t,y^j_1,y^j_2,\ldots,y^j_t,y^j_0,z^j_0,z^j_1,\ldots,z^j_{t},x^j_0)$ for~$j=1,2$,
and at the same time~$C_j=(c^j_0,\ldots,c^j_{3t+2},c^j_0)$, where~$c^j_0=x^j_0$  (see Fig.~\ref{fig:multicut2}); note the position of~$y^j_0$ between~$y^j_t$ and~$z^j_0$ in favor of a uniform adjacency to the instances later.
We set weights of each edge on the cycle as~$M$, except for three edges~$z^j_tx^j_0$,~$x^j_ty^j_1$,~$y^j_0z^j_0$,
which have weight~$M_{\infty}$.
For each~$i=1,\ldots,t$ we add to~$G'$ the graph~$G_i$ (with edges of weight~$1$), and connected~$s^i_1$ with~$y^1_i$ by an edge of weight~$M_\infty$,
and also add an edge of weight~$M_\infty$ between~$s^i_2$ and~$x^2_i$ (see Fig.~\ref{fig:multicut2}).
To the set~$\cT$ we add the following pairs:
\begin{enumerate}
  \item for each~$j=1,2$ and~$i=0,\ldots,3t+2$, add to~$\cT$ the pair~$(c^j_i,c^j_{(i+t+1)\bmod 3t+3})$,
  \item for each~$i=1,\ldots,t$ add to~$\cT$ every pair of vertices from the set~$\{s^i_3, x^1_i, y^2_i\}$.
\end{enumerate}
Finally as the target cutset size we set~$p'=6M+p$, which is polynomially bounded in~$p$.
Our constructed instance of \emulticut is~$I'=(G',\cT',p')$.

\begin{figure}[htp]
\begin{center}
\includegraphics[scale=1]{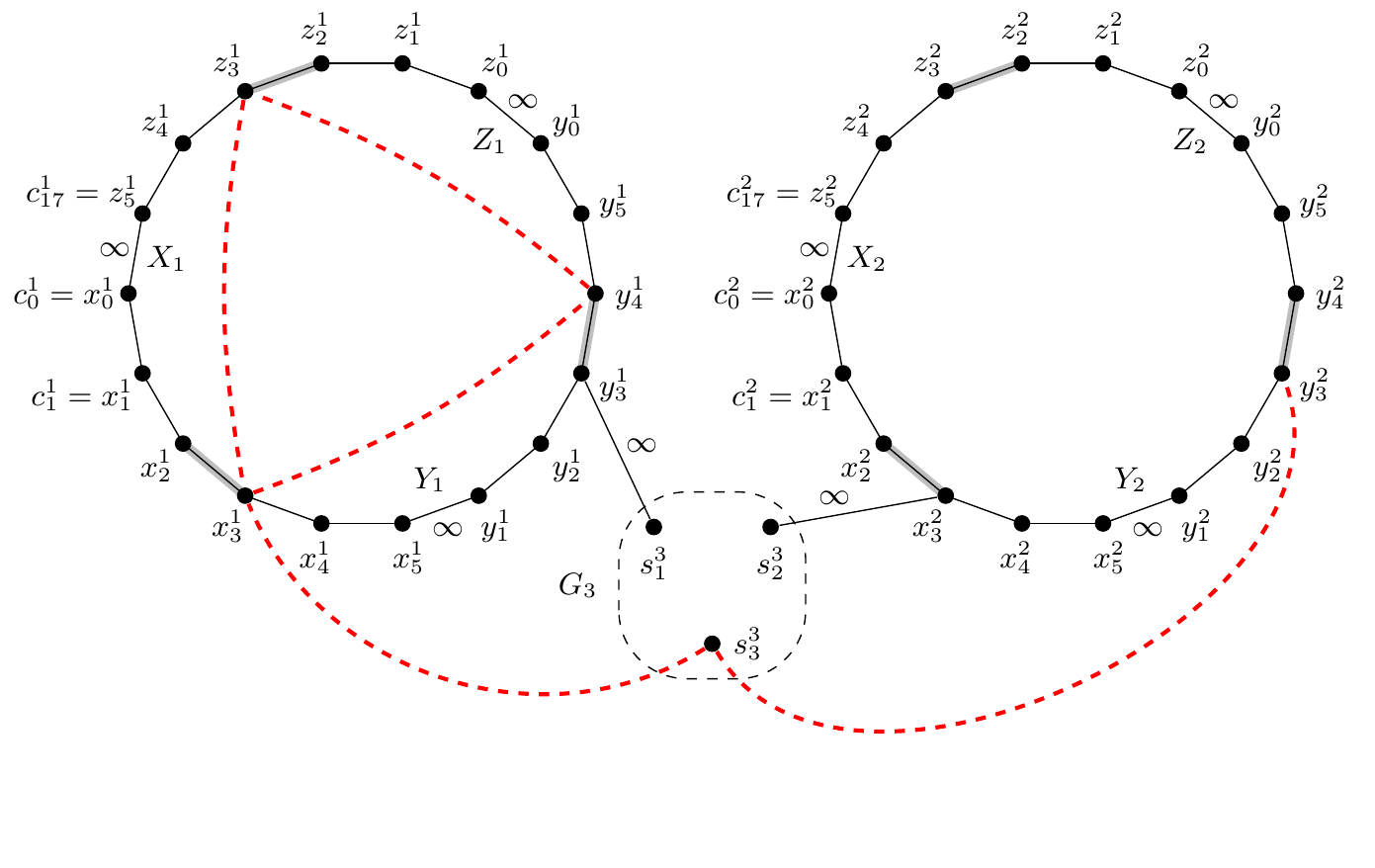}
\caption{Construction of the graph~$G'$ for~$t=5$. Dashed edges represent some of the pairs in the set~$\cT$,
in particular the pair~$(x^1_3,y^2_3) \in \cT$, but it is not depicted to simplify the figure.
Gray-marked edges belong to the set~$S'$ constructed in the proof of Theorem~\ref{thm:cross-multicut2}.}
\label{fig:multicut2}
\end{center}
\end{figure}

\paragraph{Analysis.}
First let us assume, that there is an index~$1 \le i_0 \le t$, such that~$I_{i_0}$ is a YES-instance,
and let~$S$ by any solution for~$I_{i_0}$.
We show that~$I'$ is also a YES-instance.
Let~$S':=S$ and add to~$S'$ the~$6$ edges~$c^j_{i_0-1+r(t+1)}c^j_{i_0+r(t+1)}$ for~$j=1,2$ and~$r=0,1,2$ (see gray-marked
edges in Fig.~\ref{fig:multicut2}).
Note that each of those~$6$ edges has weight exactly~$M$.
We claim that~$S'$ is a solution for~$I'$. First, observe that the total weight of edges in~$S'$
is at most~$6M+p=p'$.
Let us analyze how connected components in~$G' \setminus S'$ look like.
For~$1 \le j \le 2$, let~$X_j$,~$Y_j$,~$Z_j$ be the connected components of~$G' \setminus S'$
(not necessarily pairwise different), containing the edges~$z^j_tx^j_0$,~$x^j_ty^j_1$, and~$y^j_0z^j_0$
respectively.
Note that those are the edges of weight~$M_\infty$ and they do not belong to~$S'$.
Next consider each graph~$G_i$, for~$i=1,\ldots,t$, and see which parts of the cycles~$C_1$,~$C_2$
it connects.
For~$1 \le i < i_0$, the graph~$G_i$ connects~$Y_1$ with~$X_2$.
For~$i_0 < i \le t$, the graph~$G_i$ connects~$Z_1$ with~$Y_2$.
Finally the graph~$G_{i_0} \setminus S'$ does not connect any connected components,
since in~$G_{i_0} \setminus S'$ each of the three terminals of~$T_i$ is in a different connected component.
Therefore we know that for each~$A \in \{X_1\}, B \in \{Y_1,X_2\}, C\in\{Z_1,Y_2\}, D\in\{Z_2\}$
the connected components~$A$,~$B$,~$C$,~$D$ are pairwise different.
We analyze all the pairs in the set~$\cT$ and argue that the two elements of each pair are in different connected components in~$G' \setminus S'$.
Consider any~$j=1,2$. Note that the connected components~$X_j$,~$Y_j$,~$Z_j$ are pairwise different,
and hence for~$i=0,\ldots,3t+2$ the vertex~$c^j_i$ is in a different connected component that~$c^j_{(i+t+1)\bmod 3t+3}$ in~$G'\setminus S'$.
Therefore it is enough to analyze pairs~$(s^i_3, x^1_i), (s^i_3, y^2_i), (x^1_i, y^2_i) \in \cT$ for~$i=1,\ldots,t$.
\begin{enumerate}
  \item For~$1 \le i < i_0$, we have~$s^i_3 \in Y_1=X_2$, while~$x^1_i \in X_1$ and~$y^2_i \in Y_2=Z_1$.
  \item For~$i_0 < i \le t$, we have~$s^i_3 \in Z_1=Y_2$, while~$x^1_i \in Y_1=X_2$ and~$y^2_i \in Z_2$.
  \item For~$i=i_0$, the vertex~$s^i_3$ is in different connected component than all the vertices in the cycles~$C_1$,~$C_2$,
  since in~$G_{i_0} \setminus S'$ there are no~$s^i_3s^i_1$-paths, nor~$s^i_3s^i_2$-paths.
  At the same time~$x^1_i \in Y_1=X_2$ and~$y^2_i \in Y_2=Z_1$.
\end{enumerate}
Consequently for each pair~$(s,t) \in \cT$, vertices~$s$ and~$t$ are in different connected components of~$G' \setminus S'$,
therefore~$I'$ is a YES-instance.

In the other direction, let us assume that~$I'$ is a YES-instance and let~$S'$ be any solution for~$I'$.
Observe that out of each~$t+1$ consecutive edges on a cycle~$C_j$, for~$j=1,2$, the set~$S'$ has to contain at least one edge,
since otherwise there would be a pair of vertices~$(c^j_{i},c^j_{(i+1+1) \bmod 3t+3}) \in \cT$ belonging to the same connected component of~$G' \setminus S'$.
However,~$|C_j|=3(t+1)$ and therefore the set~$S'$ contains at least three edges of~$C_j$.
Moreover~$7M > p'$, hence~$S'$ contains exactly three equidistant edges out of each cycle~$C_j$,
since otherwise there would be~$t+1$ consecutive edges not belonging to~$S'$. 
Since~$M_\infty > p'$ there are exactly~$t$ layouts of three equidistant edges in each~$C_j$ which 
do not contain any edge of weight~$M_\infty$.
Observe that because of the way we labeled vertices on each cycle,
there exists an index~$i_j$, such that~$x^j_i$ and~$y^j_i$ are in
the same connected component of~$G' \setminus S'$, namely~$1 \le i_j \le t$,
such that~$S'$ contains the edge~$x^j_{i_j-1}x^j_{i_j}$.

We have to consider two cases, either~$i_1=i_2$ or~$i_1 \not= i_2$.
Let~$S'' \subseteq S'$ be the subset of edges of~$S'$ of weight~$1$.
Note that~$|S''| \le p$.

\textbf{Case 1 ($\boldsymbol{i_1\neq i_2}$).} 
Since~$(s^{i_1}_3,x^1_{i_1}) \in \cT$, and by the fact 
that~$x^1_{i_1}$ and~$y^1_{i_1}$ are in the same connected component of~$G'\setminus S'$,
there is no~$s^{i_1}_3s^{i_1}_1$-path in~$G_i \setminus S''$.
By properties (i), (ii), (iii) of \prepmc problem definition, we infer~$|S'' \cap E(G_{i_1})| > p/2$,
since for~$d=\deg_{G_i}(s^{i_1}_1)$ we have~$|S'' \cap E(G_{i_1})| \ge d$ and~$2d > p$.
Analogously, since~$(s^{i_2}_3,y^2_{i_2}) \in \cT$, we have~$|S'' \cap E(G_{i_2})| > p/2$, but then~$|S''| > p$, a contradiction.

\textbf{Case 2 ($\boldsymbol{i_1= i_2}$).} 
Let~$i_0=i_1=i_2$.
By the definition of~$\cT$,
each pair of vertices of~$\{s^{i_0}_3,x^1_{i_0},y^2_{i_0}\}$ belongs
to different connected component of~$G' \setminus S'$.
Observe that~$s^{i_0}_1$,~$y^1_{i_0}$,~$x^1_{i_0}$ are in the same connected component of~$G' \setminus S'$,
since the edge~$s^{i_0}_1y^1_{i_0}$ is of weight~$M_\infty$, and for this reason it does not belong to~$S'$.
Similarly~$s^{i_0}_2$,~$x^2_{i_0}$,~$y^2_{i_0}$  are in the same connected component of~$G' \setminus S'$.
Therefore there is no path between any pair of vertices of~$T_{i_0}$ in~$G_{i_0} \setminus S''$,
and since~$|S''| \le p$ we infer that~$I_{i_0}$ is a YES-instance.
\end{proof}

%% file: k-way-cut-full.tex
\section{$k$-Way Cut}\label{sec:k-way-full}

In this section we study the following graph separation problem.

\defproblemnoparam{\kwaycut}{An undirected connected
  graph $G$ and integers $k$ and $s$.}{Does there exist a set $X$ of at most $s$ edges in $G$ such that $G \setminus X$ has at least $k$ connected components?}

The \kwaycut problem, parameterized by $s$, was proven to be fixed-parameter tractable
by Kawarabayashi and Thorup \cite{ken:kwaycut}.
The problem is $W[1]$-hard when parameterized by $k$ \cite{kway:w1hard}, as well as when
we allow vertex deletions instead of edge deletions, and parameterize by $s$ \cite{marx:cut}.

Note that in the problem definition we assume that the input graph is connected and, therefore,
for $k > s+1$ the input instances are trivial. However, if we are given an instance $(G,k,s)$
where $G$ has $c > 1$ connected components, we can easily reduce it to the connected version:
we add to $G$ a complete graph on $s+2$ vertices (so that no two vertices of the complete graph
can be separated by a cut of size $s$), connect one vertex from each connected component of $G$ to all vertices of the
complete graph, and decrease $k$ by $c-1$. Thus, by restricting ourselves to connected graphs $G$ we do not make the problem easier.

The main result of this section is that \kwaycut, parameterized by $s$, does not admit a polynomial kernel
(unless \unlesscompass).
We show a cross-composition from the \clique problem, well-known to be NP-complete.

\defproblemnoparam{\clique}{An undirected graph $G$ and an integer $\ell$.}{
  Does $G$ contain a clique on $\ell$ vertices as a subgraph?}

\begin{theorem}\label{thm:kway-main-full}
\clique cross-composes to \kwaycut parameterized by $s$.
\end{theorem}

\begin{proof}
We start by defining a relation $\prel$ on \clique input instances as follows:
$(G,\ell)$ is in relation $\prel$ with $(G',\ell')$ if $\ell = \ell'$, $|V(G)| = |V(G')|$
and $|E(G)| = |E(G')|$. Clearly, $\prel$ is a polynomial equivalence relation.
Thus, in the designed cross-composition,
  we may assume that we are given $t$ instances $(G_i,\ell)$ ($1 \leq i \leq t$)
of the \clique problem and $|V(G_i)| = n$, $|E(G_i)| = m$ for all $1 \leq i \leq t$. 
Moreover, we assume that $m \geq \binom{\ell}{2}$ and $1 < \ell \leq n$,
  as otherwise all input instances are trivial.

We first consider a weighted version of the \kwaycut problem where
each edge may have a positive integer weight and the cutset $X$
needs to be of total weight at most $s$. The weights in our construction
are polynomial in $n$ and $m$. At the end we show how to reduce the weighted
version to the unweighted one.

We start by defining $k = n-\ell+1$, $w_1 = m$, $w_2 = m\binom{n}{2}$ and
$s = w_2(n-\ell) + w_1\left(\binom{n}{2} - \binom{\ell}{2}\right) + m - \binom{\ell}{2}$.
Note that $s < w_2(n-\ell+1)$ and $s < w_2(n-\ell) + w_1(\binom{n}{2} - \binom{\ell}{2}+1)$.

For each graph $G_i$, $1 \leq i \leq t$, we define a graph
$G_i'$ as a complete graph on $n$ vertices with vertex set $V(G_i)$,
  where the edge $uv$ has weight $w_1+1$ if $uv \in E(G_i)$ and
  weight $w_1$ otherwise.
We construct a graph $G$ as follows. We take
a disjoint union of all graphs $G_i'$ for $1 \leq i \leq t$,
add a root vertex $r$ and for each $1 \leq i \leq t$, $v \in V(G_i')$
we add an edge $rv$ of weight $w_2$.

Clearly $G$ is connected, $s$ is polynomial in $n$ and $m$ and the graph $G$ can be constructed
in polynomial time. We claim that $(G,k,s)$ is a weighted \kwaycut YES-instance
if and only if one of the input \clique instances $(G_i,\ell)$ is a YES-instance.

First, assume that for some $1 \leq i \leq t$, the \clique instance $(G_i,\ell)$ is
a YES-instance. Let $C \subseteq V(G_i)$ be a witness: $|C| = \ell$ and $G_i[C]$
is a clique. Consider a set $X \subseteq E(G)$ containing all edges of $G$ incident
to $V(G_i') \setminus C$. Clearly, $G \setminus X$ contains $k = n-\ell+1$ connected components:
we have one large connected component with vertex set $(V(G) \setminus V(G_i')) \cup C$
and each of $n-\ell$ vertices of $V(G_i') \setminus C$ is an isolated vertex in $G \setminus X$.
Let us now count the total weight of edges in $X$.
$X$ contains $n - \ell$ edges of weight $w_2$ that connect $V(G_i') \setminus C$ to the root $r$.
Moreover, $X$ contains $\binom{n}{2} - \binom{\ell}{2}$ edges of $G_i'$, of weight $w_1$
or $w_1+1$.
Since $G_i[C]$ is a clique, only $m - \binom{\ell}{2}$ of the edges in $X$ are of weight
$w_1+1$. Thus the total weight of edges in $X$ is equal to
$w_2(n - \ell) + w_1\left(\binom{n}{2} - \binom{\ell}{2}\right) + m - \binom{\ell}{2} = s$.

In the other direction, let $X \subseteq E(G)$ be a solution to the \kwaycut instance
$(G,k,s)$. Let $Z$ be the connected component of $G \setminus X$ that contains
the root $r$.
Let $Y \subseteq V(G)$ be the set of vertices that are not in $Z$.
If $v \in Y$, $X$ contains the edge $rv$ of weight $w_2$.
As $s < w_2(n-\ell+1)$, we have $|Y| \leq n-\ell$.
As $k = n - \ell+1$, we infer that $G \setminus X$ contains $n-\ell+1$ connected components:
$Z$ and $n-\ell$ isolated vertices.
That is, $|Y| = n - \ell$ and all vertices in $Y$ are isolated in $G \setminus X$.
Note that $X$ includes $n-\ell$ edges of weight $w_2$ that connect the root $r$ with the vertices of $Y$.

The next step is to prove that all vertices of $Y$ are contained in one of the graphs $G_i'$.
To this end, let $a_i = |Y \cap V(G_i')|$ for $1 \leq i \leq t$.
Note that $X \cap E(G_i')$ contains at least $\binom{a_i}{2} + a_i(n-a_i)$ edges of weight $w_1$ or $w_1+1$.
Thus, the number of edges of weight $w_1$ or $w_1+1$ contained in $X$ is at least:
\begin{align*}
&\sum_{i=1}^t \binom{a_i}{2} + a_i(n-a_i) = \left(n-\frac{1}{2}\right)\sum_{i=1}^t a_i - \frac{1}{2} \sum_{i=1}^t a_i^2 = (n-\ell)\left(n-\frac{1}{2}\right) - \frac{1}{2} \sum_{i=1}^t a_i^2 \\
    &\qquad \geq (n-\ell)\left(n-\frac{1}{2}\right) -\frac{1}{2} \left(\sum_{i=1}^t a_i\right)^2 = (n-\ell)\left(n-\frac{1}{2}\right) - \frac{1}{2} (n-\ell)^2 = \binom{n}{2} - \binom{\ell}{2}.
\end{align*}
As $s < w_2(n-\ell) + w_1\left(\binom{n}{2} - \binom{\ell}{2}+1\right)$, we infer that the number of edges in $X$ of weight $w_1$ or $w_1+1$ is exactly $\binom{n}{2} - \binom{\ell}{2}$.
This is only possible if $\sum_{i=1}^t a_i^2 = (\sum_{i=1}^t a_i)^2$. As $a_i$ are nonnegative integers, we infer that only one value $a_i$ is positive.

Thus $Y \subseteq V(G_i')$ for some $1 \leq i \leq t$. Let $C = V(G_i) \setminus Y$. Note that $|C| = \ell$.
The set $X$ contains all $\binom{n}{2} - \binom{\ell}{2}$ edges of $G_i'$ that are incident to $Y$.
As the total weight of the edges of $X$ is at most $s$, $X$ contains at most $m - \binom{\ell}{2}$ edges of weight $w_1+1$.
We infer that there are at least $\binom{\ell}{2}$ edges in the graph $G_i[C]$, $G_i[C]$ is a clique and $(G_i,\ell)$ is a YES-instance of the \clique problem.

To finish the proof, we show how to reduce the weighted version of the \kwaycut problem to the unweighted one.
We replace each vertex $u$ with a complete graph $H_u$ on $s+2$ vertices
and for each edge $uv$ of weight $w$ we add to the graph $w$ arbitrarily chosen edges between $H_u$ and $H_v$ (note that in our construction all weights are smaller than $s$).
Note that this reduction preserves the connectivity of the graph $G$.
Let $X$ be a solution to the unweighted instance $(G,k,s)$ constructed in this way.
As no cut of size at most $s$ can separate two vertices of $H_u$, each clique $H_u$ is contained in one connected component of $G \setminus X$.
Moreover, to separate $H_u$ from $H_v$, $X$ needs to include all $w$ edges between $H_u$ and $H_w$. Thus, the constructed unweighted instance is
indeed equivalent to the weighted one.
Note that in the presented cross-composition the edge weights were polynomial in $n$ and $m$, so the presented reduction can be performed in polynomial time.
\end{proof}

By applying Theorem \ref{thm:cross-composition} we obtain the following corollary.
\begin{corollary}
\kwaycut parameterized by $s$ does not admit a polynomial kernel unless \unlesscompass.
\end{corollary}

%% file: conclusions.tex
\section{Conclusion and open problems}\label{sec:conc}

We have shown that four important parameterized problems do not admit a kernelization
algorithm with a polynomial guarantee on the output size unless \unlesscompass and the polynomial
hierarchy collapses. 
We would like to mention here a few open problems very closely related
to our work.
\begin{itemize}
\item The $2^k$-vertex kernel for \cliquecover \cite{gghn:cc} is probably close to
optimal. Currently the fastest fixed-parameter algorithm for \cliquecover is a brute-force
algorithm on the exponential kernel. Is this double-exponential dependency on $k$ necessary?
\item The OR-composition for \dirmc in the case of two terminals excludes
the existence of a polynomial kernel for most graph separation problems in directed graphs.
There are two important cases not covered by this result: one is the \multicut problem
in directed acyclic graphs, and the second is \dirmc with deletable terminals.
To the best of our knowledge, it is also open whether the first problem is fixed-parameter tractable.
\item Both our OR-compositions for \multicut use a number of terminal pairs that is linear in the number of input instances.
Is \multicut parameterized by both the size of the cutset and the number of terminal pairs similarly hard to kernelize?
\end{itemize}

%% file: no-colours-and-ids.bbl
\begin{thebibliography}{10}

\bibitem{hitting-set:kernel}
Faisal~N. Abu-Khzam.
\newblock A kernelization algorithm for d-hitting set.
\newblock {\em J. Comput. Syst. Sci.}, 76(7):524--531, 2010.

\bibitem{DBLP:conf/icalp/2011-1}
Luca Aceto, Monika Henzinger, and Jiri Sgall, editors.
\newblock {\em Automata, Languages and Programming - 38th International
  Colloquium, ICALP 2011, Zurich, Switzerland, July 4-8, 2011, Proceedings,
  Part I}, volume 6755 of {\em Lecture Notes in Computer Science}. Springer,
  2011.

\bibitem{cc-apl1}
Pankaj~K. Agarwal, Noga Alon, Boris Aronov, and Subhash Suri.
\newblock Can visibility graphs be represented compactly?
\newblock {\em Discrete {\&} Computational Geometry}, 12:347--365, 1994.

\bibitem{apx:cc}
G.~Ausiello, P.~Crescenzi, G.~Gambosi, V.~Kann, A.~Marchetti-Spaccamela, and
  M.~Protasi.
\newblock {\em Complexity and Approximation: Combinatorial Optimization
  Problems and Their Approximability Properties}.
\newblock Springer, 1999.

\bibitem{bt:cc}
Michael Behrisch and Anusch Taraz.
\newblock Efficiently covering complex networks with cliques of similar
  vertices.
\newblock {\em Theor. Comput. Sci.}, 355(1):37--47, 2006.

\bibitem{Bodlaender09}
Hans~L. Bodlaender.
\newblock Kernelization: New upper and lower bound techniques.
\newblock In Chen and Fomin \cite{DBLP:conf/iwpec/2009}, pages 17--37.

\bibitem{bodlaender:kernel}
Hans~L. Bodlaender, Rodney~G. Downey, Michael~R. Fellows, and Danny Hermelin.
\newblock On problems without polynomial kernels.
\newblock {\em J. Comput. Syst. Sci.}, 75(8):423--434, 2009.

\bibitem{meta-kernelization}
Hans~L. Bodlaender, Fedor~V. Fomin, Daniel Lokshtanov, Eelko Penninkx, Saket
  Saurabh, and Dimitrios~M. Thilikos.
\newblock ({M}eta) kernelization.
\newblock In {\em FOCS}, pages 629--638. IEEE Computer Society, 2009.

\bibitem{cross-composition}
Hans~L. Bodlaender, Bart M.~P. Jansen, and Stefan Kratsch.
\newblock Cross-composition: A new technique for kernelization lower bounds.
\newblock In Thomas Schwentick and Christoph D{\"u}rr, editors, {\em STACS},
  volume~9 of {\em LIPIcs}, pages 165--176. Schloss Dagstuhl - Leibniz-Zentrum
  fuer Informatik, 2011.

\bibitem{bjk:treewidth}
Hans~L. Bodlaender, Bart M.~P. Jansen, and Stefan Kratsch.
\newblock Preprocessing for treewidth: A combinatorial analysis through
  kernelization.
\newblock In Aceto et~al. \cite{DBLP:conf/icalp/2011-1}, pages 437--448.

\bibitem{bodlaender:transformations}
Hans~L. Bodlaender, S.~Thomasse, and A.~Yeo.
\newblock Analysis of data reduction: {T}ransformations give evidence for
  non-existence of polynomial kernels, 2008.
\newblock Technical Report UU-CS-2008-030, Institute of Information and
  Computing Sciences, Utrecht University, Netherlands.

\bibitem{multicut-fpt1}
Nicolas Bousquet, Jean Daligault, and St{\'e}phan Thomass{\'e}.
\newblock Multicut is {FPT}.
\newblock In Fortnow and Vadhan \cite{DBLP:conf/stoc/2011}, pages 459--468.

\bibitem{kway:motiv}
Michel Burlet and Olivier Goldschmidt.
\newblock A new and improved algorithm for the 3-cut problem.
\newblock {\em Oper. Res. Lett.}, 21(5):225--227, 1997.

\bibitem{buss-goldsmith}
Jonathan~F. Buss and Judy Goldsmith.
\newblock Nondeterminism within {P}.
\newblock {\em SIAM J. Comput.}, 22(3):560--572, 1993.

\bibitem{nokernel-collapse2}
J.~Cai, Venkatesan~T. Chakaravarthy, Lane~A. Hemaspaandra, and Mitsunori
  Ogihara.
\newblock Competing provers yield improved {K}arp-{L}ipton collapse results.
\newblock {\em Inf. Comput.}, 198(1):1--23, 2005.

\bibitem{mc-apx-15}
Gruia Calinescu, Howard~J. Karloff, and Yuval Rabani.
\newblock An improved approximation algorithm for multiway cut.
\newblock {\em J. Comput. Syst. Sci.}, 60(3):564--574, 2000.

\bibitem{chang-muller:cc}
Maw-Shang Chang and Haiko M{\"u}ller.
\newblock On the tree-degree of graphs.
\newblock In Andreas Brandst{\"a}dt and Van~Bang Le, editors, {\em WG}, volume
  2204 of {\em Lecture Notes in Computer Science}, pages 44--54. Springer,
  2001.

\bibitem{DBLP:conf/iwpec/2009}
Jianer Chen and Fedor~V. Fomin, editors.
\newblock {\em Parameterized and Exact Computation, 4th International Workshop,
  IWPEC 2009, Copenhagen, Denmark, September 10-11, 2009, Revised Selected
  Papers}, volume 5917 of {\em Lecture Notes in Computer Science}. Springer,
  2009.

\bibitem{chen:mc}
Jianer Chen, Yang Liu, and Songjian Lu.
\newblock An improved parameterized algorithm for the minimum node multiway cut
  problem.
\newblock {\em Algorithmica}, 55(1):1--13, 2009.

\bibitem{dfvs}
Jianer Chen, Yang Liu, Songjian Lu, Barry O'Sullivan, and Igor Razgon.
\newblock A fixed-parameter algorithm for the directed feedback vertex set
  problem.
\newblock {\em J. ACM}, 55(5), 2008.

\bibitem{marx:dir-mc}
Rajesh Chitnis, Mohammadtaghi Hajiaghayi, and D\'{a}niel Marx.
\newblock Fixed-parameter tractability of directed multiway cut parameterized
  by the size of the cutset.
\newblock In {\em SODA (to appear)}, 2012.

\bibitem{nasze:eulerian}
Marek Cygan, D\'{a}niel Marx, Marcin Pilipczuk, Michal Pilipczuk, and
  Ildik\'{o} Schlotter.
\newblock Parameterized complexity of eulerian deletion problems.
\newblock In {\em WG (to appear)}, 2011.

\bibitem{moje:povd}
Marek Cygan, Marcin Pilipczuk, Michal Pilipczuk, and Jakub~Onufry Wojtaszczyk.
\newblock An improved {FPT} algorithm and quadratic kernel for pathwidth one
  vertex deletion.
\newblock In Venkatesh Raman and Saket Saurabh, editors, {\em IPEC}, volume
  6478 of {\em Lecture Notes in Computer Science}, pages 95--106. Springer,
  2010.

\bibitem{nasze:wg10}
Marek Cygan, Marcin Pilipczuk, Michal Pilipczuk, and Jakub~Onufry Wojtaszczyk.
\newblock Kernelization hardness of connectivity problems in {\it d}-degenerate
  graphs.
\newblock In Dimitrios~M. Thilikos, editor, {\em WG}, volume 6410 of {\em
  Lecture Notes in Computer Science}, pages 147--158, 2010.

\bibitem{nmc:2k}
Marek Cygan, Marcin Pilipczuk, Micha\l{} Pilipczuk, and Jakub~Onufry
  Wojtaszczyk.
\newblock On multiway cut parameterized above lower bounds.
\newblock In {\em IPEC (to appear)}, 2011.
\newblock Available at http://arxiv.org/abs/1107.1585.

\bibitem{sfvs}
Marek Cygan, Marcin Pilipczuk, Michal Pilipczuk, and Jakub~Onufry Wojtaszczyk.
\newblock Subset feedback vertex set is fixed-parameter tractable.
\newblock In Aceto et~al. \cite{DBLP:conf/icalp/2011-1}, pages 449--461.

\bibitem{nmc-np-hardness}
Elias Dahlhaus, David~S. Johnson, Christos~H. Papadimitriou, Paul~D. Seymour,
  and Mihalis Yannakakis.
\newblock The complexity of multiterminal cuts.
\newblock {\em SIAM J. Comput.}, 23(4):864--894, 1994.

\bibitem{dailey}
David~P. Dailey.
\newblock Uniqueness of colorability and colorability of planar 4-regular
  graphs are {NP}-complete.
\newblock {\em Discrete Mathematics}, 30(3):289–293, 1980.

\bibitem{dell-marx:soda12}
Holger Dell and D\'{a}niel Marx.
\newblock Kernelization of packing problems.
\newblock In {\em SODA (to appear)}, 2012.

\bibitem{dell:kernel}
Holger Dell and Dieter van Melkebeek.
\newblock Satisfiability allows no nontrivial sparsification unless the
  polynomial-time hierarchy collapses.
\newblock In Leonard~J. Schulman, editor, {\em STOC}, pages 251--260. ACM,
  2010.

\bibitem{colours-and-ids}
Michael Dom, Daniel Lokshtanov, and Saket Saurabh.
\newblock Incompressibility through colors and {ID}s.
\newblock In Susanne Albers, Alberto Marchetti-Spaccamela, Yossi Matias,
  Sotiris~E. Nikoletseas, and Wolfgang Thomas, editors, {\em ICALP (1)}, volume
  5555 of {\em Lecture Notes in Computer Science}, pages 378--389. Springer,
  2009.

\bibitem{kway:w1hard}
Rodney~G. Downey, Vladimir Estivill-Castro, Michael~R. Fellows, Elena Prieto,
  and Frances~A. Rosamond.
\newblock Cutting up is hard to do: the parameterized complexity of k-cut and
  related problems.
\newblock {\em Electr. Notes Theor. Comput. Sci.}, 78:209--222, 2003.

\bibitem{cliquecover:erdos}
Paul Erd\"{o}s, A.~W. Goodman, and Lajos Posa.
\newblock The representation of a graph by set intersections.
\newblock {\em Canadian Journal of Mathematics}, 18:106--112, 1966.

\bibitem{fernau:outtrees}
Henning Fernau, Fedor~V. Fomin, Daniel Lokshtanov, Daniel Raible, Saket
  Saurabh, and Yngve Villanger.
\newblock Kernel(s) for problems with no kernel: On out-trees with many leaves.
\newblock In Susanne Albers and Jean-Yves Marion, editors, {\em STACS},
  volume~3 of {\em LIPIcs}, pages 421--432. Schloss Dagstuhl - Leibniz-Zentrum
  fuer Informatik, Germany, 2009.

\bibitem{fomin:bidim-kernels}
Fedor~V. Fomin, Daniel Lokshtanov, Saket Saurabh, and Dimitrios~M. Thilikos.
\newblock Bidimensionality and kernels.
\newblock In Moses Charikar, editor, {\em SODA}, pages 503--510. SIAM, 2010.

\bibitem{fortnow:kernel}
Lance Fortnow and Rahul Santhanam.
\newblock Infeasibility of instance compression and succinct {PCP}s for {NP}.
\newblock {\em J. Comput. Syst. Sci.}, 77(1):91--106, 2011.

\bibitem{DBLP:conf/stoc/2011}
Lance Fortnow and Salil~P. Vadhan, editors.
\newblock {\em Proceedings of the 43rd ACM Symposium on Theory of Computing,
  STOC 2011, San Jose, CA, USA, 6-8 June 2011}. ACM, 2011.

\bibitem{garey-johnson}
M.~R. Garey and D.~S. Johnson.
\newblock {\em Computers and Intractability: A Guide to the Theory of
  {NP}-Completeness}.
\newblock New York: W.H. Freeman, 1979.

\bibitem{multicut-apx}
Naveen Garg, Vijay~V. Vazirani, and Mihalis Yannakakis.
\newblock Approximate max-flow min-(multi)cut theorems and their applications.
\newblock {\em SIAM J. Comput.}, 25(2):235--251, 1996.

\bibitem{mc-apx-2}
Naveen Garg, Vijay~V. Vazirani, and Mihalis Yannakakis.
\newblock Multiway cuts in node weighted graphs.
\newblock {\em J. Algorithms}, 50(1):49--61, 2004.

\bibitem{kway:first}
O.~Goldschmidt and D.~S. Hochbaum.
\newblock A polynomial algorithm for the $k$-cut problem for fixed $k$.
\newblock {\em Math. Oper. Res.}, 19(1):24--37, 1994.

\bibitem{gghn:cc}
Jens Gramm, Jiong Guo, Falk H{\"u}ffner, and Rolf Niedermeier.
\newblock Data reduction and exact algorithms for clique cover.
\newblock {\em ACM Journal of Experimental Algorithmics}, 13, 2008.

\bibitem{cc-apl2}
Jens Gramm, Jiong Guo, Falk H{\"u}ffner, Rolf Niedermeier, Hans-Peter Piepho,
  and Ramona Schmid.
\newblock Algorithms for compact letter displays: Comparison and evaluation.
\newblock {\em Computational Statistics {\&} Data Analysis}, 52(2):725--736,
  2007.

\bibitem{cliquecover:book}
Jonathan~L. Gross and Jay Yellen.
\newblock {\em Graph Theory and its Applications}.
\newblock CRC Press, 2006.

\bibitem{guillaume:cc}
Jean-Loup Guillaume and Matthieu Latapy.
\newblock Bipartite structure of all complex networks.
\newblock {\em Inf. Process. Lett.}, 90(5):215--221, 2004.

\bibitem{Guillemot11a}
Sylvain Guillemot.
\newblock {FPT} algorithms for path-transversal and cycle-transversal problems.
\newblock {\em Discrete Optimization}, 8(1):61--71, 2011.

\bibitem{guo:survey}
Jiong Guo and Rolf Niedermeier.
\newblock Invitation to data reduction and problem kernelization.
\newblock {\em SIGACT News}, 38(1):31--45, 2007.

\bibitem{HarnikN10}
Danny Harnik and Moni Naor.
\newblock On the compressibility of {{\it NP}} instances and cryptographic
  applications.
\newblock {\em SIAM J. Comput.}, 39(5):1667--1713, 2010.

\bibitem{hermelin-wu:soda12}
Danny Hermelin and Xi~Wu.
\newblock Weak compositions and their applications to polynomial lower-bounds
  for kernelization.
\newblock In {\em SODA (to appear)}, 2012.

\bibitem{hoover:cc}
D.~N. Hoover.
\newblock Complexity of graph covering problems for graphs of low degree.
\newblock {\em Journal of Combinatorial Mathematics and Combinatorial
  Computing}, 11:187--208, 1992.

\bibitem{cc-class2}
Wen-Lian Hsu and Kuo-Hui Tsai.
\newblock Linear time algorithms on circular-arc graphs.
\newblock {\em Inf. Process. Lett.}, 40(3):123--129, 1991.

\bibitem{kway:improve2}
Yoko Kamidoi, Noriyoshi Yoshida, and Hiroshi Nagamochi.
\newblock A deterministic algorithm for finding all minimum k-way cuts.
\newblock {\em SIAM J. Comput.}, 36(5):1329--1341, 2007.

\bibitem{mc-apx-138}
David~R. Karger, Philip~N. Klein, Clifford Stein, Mikkel Thorup, and Neal~E.
  Young.
\newblock Rounding algorithms for a geometric embedding of minimum multiway
  cut.
\newblock {\em Math. Oper. Res.}, 29(3):436--461, 2004.

\bibitem{kway:improve1}
David~R. Karger and Clifford Stein.
\newblock A new approach to the minimum cut problem.
\newblock {\em J. ACM}, 43(4):601--640, 1996.

\bibitem{ken:kwaycut}
{Ken-ichi} Kawarabayashi and Mikkel Thorup.
\newblock Minimum k-way cut of bounded size is fixed-parameter tractable.
\newblock In {\em FOCS (to appear)}, 2011.

\bibitem{kellerman}
E.~Kellerman.
\newblock Determination of keyword conflict.
\newblock {\em IBM Technical Disclosure Bulletin}, 16(2):544--546, 1973.

\bibitem{kou:cc}
L.~T. Kou, L.~J. Stockmeyer, and C.-K. Wong.
\newblock Covering edges by cliques with regard to keyword conflicts and
  intersection graphs.
\newblock {\em Communications of the ACM}, 21(2):135--139, 1978.

\bibitem{stefan:ramsey}
Stefan Kratsch.
\newblock Co-nondeterminism in compositions: A kernelization lower bound for a
  ramsey-type problem.
\newblock In {\em SODA (to appear)}, 2012.

\bibitem{stefan:two}
Stefan Kratsch and Magnus Wahlstr{\"o}m.
\newblock Two edge modification problems without polynomial kernels.
\newblock In Chen and Fomin \cite{DBLP:conf/iwpec/2009}, pages 264--275.

\bibitem{stefan:new}
Stefan Kratsch and Magnus Wahlstr{\"o}m.
\newblock Representative sets and irrelevant vertices: New tools for
  kernelization, 2011.
\newblock Unpublished manuscript.

\bibitem{stefan:oct}
Stefan Kratsch and Magnus Wahlstr{\"o}m.
\newblock Compression via matroids: A randomized polynomial kernel for odd
  cycle transversal.
\newblock In {\em SODA (to appear)}, 2012.

\bibitem{lund-yannakakis}
Carsten Lund and Mihalis Yannakakis.
\newblock On the hardness of approximating minimization problems.
\newblock {\em J. ACM}, 41(5):960--981, 1994.

\bibitem{cc-class1}
S.~Ma, W.~D. Wallis, and J.~Wu.
\newblock Clique covering of chordal graphs.
\newblock {\em Utilitas Mathematica}, 36:151--152, 1989.

\bibitem{marx:cut}
D{\'a}niel Marx.
\newblock Parameterized graph separation problems.
\newblock {\em Theor. Comput. Sci.}, 351(3):394--406, 2006.

\bibitem{multicut-fpt2}
D{\'a}niel Marx and Igor Razgon.
\newblock Fixed-parameter tractability of multicut parameterized by the size of
  the cutset.
\newblock In Fortnow and Vadhan \cite{DBLP:conf/stoc/2011}, pages 469--478.

\bibitem{dir-mc-apx}
Joseph Naor and Leonid Zosin.
\newblock A 2-approximation algorithm for the directed multiway cut problem.
\newblock {\em SIAM J. Comput.}, 31(2):477--482, 2001.

\bibitem{vc:2k}
George~L. Nemhauser and Leslie~E. Trotter.
\newblock Vertex packings: Structural properties and algorithms.
\newblock {\em Math. Program.}, 8:232--248, 1975.

\bibitem{orlin:cc}
J.~B. Orlin.
\newblock Contentment in graph theory: Covering graphs with cliques.
\newblock {\em Indagationes Mathematicae (Proceedings)}, 80(5):406--424, 1977.

\bibitem{cc-apl3}
Hans-Peter Piepho.
\newblock An algorithm for a letter-based representation of all-pairwise
  comparisons.
\newblock {\em Journal of Computational and Graphical Statistics},
  13(2):456--466, 2004.

\bibitem{cc-apl4}
Subramanian Rajagopalan, Manish Vachharajani, and Sharad Malik.
\newblock Handling irregular {ILP} within conventional {VLIW} schedulers using
  artificial resource constraints.
\newblock In {\em CASES}, pages 157--164, 2000.

\bibitem{razgon:arxiv2010}
Igor Razgon.
\newblock Computing multiway cut within the given excess over the largest
  minimum isolating cut.
\newblock {\em CoRR}, abs/1011.6267, 2010.

\bibitem{razgon:arxiv2011}
Igor Razgon.
\newblock Large isolating cuts shrink the multiway cut.
\newblock {\em CoRR}, abs/1104.5361, 2011.

\bibitem{almost2sat-fpt}
Igor Razgon and Barry O'Sullivan.
\newblock Almost 2-{SAT} is fixed-parameter tractable.
\newblock {\em J. Comput. Syst. Sci.}, 75(8):435--450, 2009.

\bibitem{reed:ic}
Bruce~A. Reed, Kaleigh Smith, and Adrian Vetta.
\newblock Finding odd cycle transversals.
\newblock {\em Oper. Res. Lett.}, 32(4):299--301, 2004.

\bibitem{cliquecover:appl}
Fred~S. Roberts.
\newblock Applications of edge coverings by cliques.
\newblock {\em Discrete Applied Mathematics}, 10(1):93 -- 109, 1985.

\bibitem{fvs:quadratic-kernel}
St{\'e}phan Thomass{\'e}.
\newblock A $4k^2$ kernel for feedback vertex set.
\newblock {\em ACM Transactions on Algorithms}, 6(2), 2010.

\bibitem{kway:thorup}
Mikkel Thorup.
\newblock Minimum k-way cuts via deterministic greedy tree packing.
\newblock In Cynthia Dwork, editor, {\em STOC}, pages 159--166. ACM, 2008.

\bibitem{brooks}
Helge Tverherg.
\newblock On {B}rooks' theorem and some related results.
\newblock {\em Math. Scand.}, 52:37--40, 1983.

\bibitem{xiao:multiway2010}
Mingyu Xiao.
\newblock Simple and improved parameterized algorithms for multiterminal cuts.
\newblock {\em Theory Comput. Syst.}, 46(4):723--736, 2010.

\bibitem{nokernel-collapse1}
Chee-Keng Yap.
\newblock Some consequences of non-uniform conditions on uniform classes.
\newblock {\em Theor. Comput. Sci.}, 26:287--300, 1983.

\end{thebibliography}
